\documentclass{jair}

\setcopyright{cc}
\copyrightyear{2026}
\acmYear{2026}
\acmDOI{10.1613/jair.1.19325}

\JAIRAE{Roni Stern}
\JAIRTrack{Multi-Agent Path Finding}
\acmVolume{85}
\acmArticle{36}
\acmMonth{03}
\acmYear{2026}

\newcommand{\MAPF}{\textsc{Multiagent Path Finding}\xspace}
\newcommand{\MAPFCC}{\textsc{Multiagent Path Finding with Communication Constraints}\xspace}
\newcommand{\MAPFShort}{\textsc{MAPF}\xspace}
\newcommand{\MAPFCCShort}{\textsc{MAPFCC}\xspace}
\newcommand{\MSO}{\textsf{MSO}\xspace}
\newcommand{\MSOtwo}{\textsf{MSO\textsubscript{2}}\xspace}
\newcommand{\tw}{\ensuremath{\operatorname{tw}}}
\newcommand{\N}{\mathbb{N}}

\newtheorem{observation}{Observation}

\newtheorem{theorem}{Theorem}
\newtheorem{corollary}{Corollary}
\newtheorem{claim}{Claim}
\newenvironment{proofclaim}{\noindent{\em Proof of the claim.}}{\qedclaim}
\newcommand{\qedclaim}{\hfill $\diamond$ \medskip}

\newcommand{\intzero}[1]{\{0, \dots, #1\}}
\DeclareMathOperator{\inc}{inc}

\usepackage{amsthm,amsmath}
\usepackage{mathtools}
\usepackage{complexity}
\usepackage{tikz,pgf}
\usetikzlibrary{positioning}
\usetikzlibrary{calc}
\usepackage{graphicx}
\usepackage{subfig}
\usepackage{natbib}
\usetikzlibrary{decorations.pathreplacing,calc}
\tikzset{%
  middle dotted line/.style={
    decoration={show path construction,
      lineto code={
          \draw[#1] (\tikzinputsegmentfirst) --($(\tikzinputsegmentfirst)!.3333!(\tikzinputsegmentlast)$);,
          \draw[dotted,#1] ($(\tikzinputsegmentfirst)!.3333!(\tikzinputsegmentlast)$)--($(\tikzinputsegmentfirst)!.6666!(\tikzinputsegmentlast)$);,
          \draw[#1] ($(\tikzinputsegmentfirst)!.6666!(\tikzinputsegmentlast)$)--(\tikzinputsegmentlast);,
      }
    },
    decorate
  },
}

\usepackage[obeyFinal]{todonotes}

\newcommand{\dist}{\ensuremath{\operatorname{dist}}}
\newcommand{\scheduleLength}{\ensuremath{\mu}}

\usepackage{mdframed}
\usepackage{xcolor}

\begin{document}

\title[Exact Algorithms for Multiagent Path Finding with Communication Constraints]{Exact Algorithms for Multiagent Path Finding with Communication Constraints on Tree-Like Structures}

\author{Foivos Fioravantes}
\orcid{0000-0001-8217-030X}
\email{foivos.fioravantes@fit.cvut.cz}

\author{Dušan Knop}
\orcid{0000-0003-2588-5709}
\email{dusan.knop@fit.cvut.cz}

\author{Jan Matyáš Křišťan}
\orcid{0000-0001-6657-0020}
\email{kristja6@fit.cvut.cz}

\author{Nikolaos Melissinos}
\orcid{0000-0002-0864-9803}
\email{nik.melissinos@gmail.com}

\author{Michal Opler}
\orcid{0000-0002-4389-5807}
\email{michal.opler@fit.cvut.cz}

\affiliation{%
  \institution{Department of Theoretical Computer Science, Faculty of Information Technology, Czech Technical University in Prague}
  \city{Prague}
  \country{Czech Republic}
}


\renewcommand{\shortauthors}{Fioravantes, Knop, Křišťan, Melissinos \& Opler}

\begin{abstract}
Consider the scenario where multiple agents have to move in an optimal way through a network, each one towards their ending position while avoiding collisions. By optimal, we mean as fast as possible, which is evaluated by a measure known as the makespan of the proposed solution. This is the setting studied in the \MAPF problem. In this work, we additionally provide the agents with a way to communicate with each other. Due to size constraints, it is reasonable to assume that the range of communication of each agent will be limited. What should be the trajectories of the agents to, additionally, maintain a backbone of communication? In this work, we study the \MAPFCC problem under the parameterized complexity framework.

Our main contribution is three exact algorithms that are efficient when considering particular structures for the input network. We provide such algorithms for the case when the communication range and the number of agents (the makespan resp.) are provided in the input and the network has a tree topology, or bounded maximum degree (has a tree-like topology, i.e., bounded treewidth resp.). We complement these results by showing that it is highly unlikely to construct efficient algorithms when considering the number of agents as part of the input, even if the makespan is $3$ and the communication range is $1$.
\end{abstract}


\received{30 Mai 2025}
\received[accepted]{9 March 2026}

\maketitle

\section{Introduction}
The \MAPF{} (\MAPFShort for short) problem is a well-known challenge in the field of planning and coordination.
It involves navigating multiple agents through a topological space, often modeled as an undirected graph, to reach their respective destinations.
In many real-world scenarios, additional constraints on the agents' movements are required.
One such constraint is the \emph{communication} constraint, which requires agents to maintain a connected set of vertices in a communication graph as they move; this is then the \MAPFCC{} (\MAPFCCShort for short) problem.
This requirement can arise, for example, from the need to constantly communicate with a human operator~\cite{amigoni2017multirobot}.
Sometimes, only a periodic connection might be sufficient~\cite{hollinger2012multirobot}.
On the other hand, applications in a video game movement of agents~\cite{SnapeGBLM12} should require near-connectivity, since we want the group of virtual soldiers to move in a mob.

It should also be noted that the communication constraints we consider are born as a natural first step towards further understanding and providing new insights into solving the \MAPFShort problem in the distributed setting. We believe that such a setting, where each agent needs to do some local computation taking into account only a partial view of the network and the subset of the other agents that are within its communication range, is rather natural and worth investigating. Such a framework is particularly well-suited for exploring the emergence of swarm intelligence through agent cooperation.

The complexity of the \MAPF{} problem increases significantly when the movement and communication graphs are independent of each other.
In fact, under these conditions, the problem is \PSPACE-complete~\cite{tateo2018multiagent}.
This raises a natural question: Is the problems' complexity equally severe when the movement and communication graphs are related?
For instance, if we assume that communication among agents occurs within the same space they are navigating, it is reasonable to model the communication graph as identical to the movement graph.
Alternatively, we could consider scenarios where the communication graph is a derivative of the movement graph, such as its third power, allowing agents to communicate over a distance of three edges in the original graph.
However, the problem stays \PSPACE-complete even if the agents move in a subgraph of a 3D grid and the communication is based on radius~\cite{calviac2023improved}.
We refer the reader to the next section for the formal definitions.

Both \MAPF{} and \MAPFCC{} problems are systematically studied; most researchers deal with the hardness using specific algorithms or heuristics.
The most popular approaches to find optimal solutions are using the A* algorithm (e.g. \cite{sharon2015conflict,sharon2013increasing}) or ILP solvers (e.g. \cite{yu2013planning}).
Another popular line of research used the Picat language~\cite{zhou2015constraint,BartakZSBS17}.
A wide range of heuristics is commonly used, such as those based on local search (WHCA*, ECBS), SAT solvers (e.g. \cite{SurynekSBF22}), or reinforcement learning (e.g. \cite{guptaEK2017}) to name just a few.
The WHCA* (Windowed Hierarchical Cooperative A*) approach was described and analyzed by Silver~\cite{silver2005cooperative,korf1990real}.
The ECBS (Enhanced Conflict-Based Search) approach was used by \citet{barer2014suboptimal}; similarly for the Improved CBS~\cite{BoyarskiFSSBTS15}.
For more related references, the reader might visit some of the more recent surveys on this subfield~\cite{FelnerSSBGSSWS17,surynek2022survey,SternSFK0WLA0KB19manysurvey}.

Since our paper deals with \MAPFCC{} from a theoretical point of view, we are mostly interested in the computational complexity study of it.
The study of similar-nature problems started long ago~\cite{Wilson74,KornhauserMS84,Goldreich2011} and was mostly related to puzzle games.
Many of these games were shown to be \PSPACE-complete~\cite{hearn2005pspace}.
\citet{surynek2010optimization} provided a direct proof that \MAPF{} is \NP-hard.
We stress here that the most ``direct'' argument for \NP-membership (i.e., by providing a solution) does not often work for \MAPFCCShort{}-alike problems since some agents might need to revisit some vertices in every optimal solution.
Consequently, solutions that minimize the makespan may require superpolynomial time.
Last but not least, \citet{EibenGK23,FKKMO24} studied the parameterized complexity of \MAPFShort{} and provided initial results for tree-like topology~$G$.

\paragraph{Our Contributions.}
We study the complexity of the \MAPFCCShort{} problem from the viewpoint of parameterized algorithms~\cite{downey2012parameterized}.
As the main parameter, we select the number of agents $k$ and prove that \MAPFCCShort{} is \W[1]-hard even if the makespan $\ell$ is~3, and the communication range $d$ is~1 (on the input graph); see Theorem~\ref{thm:hard-k-l-d}. In particular, this means that any single parameter among $k$, $\ell$ and/or $d$ is highly unlikely to lead to an \FPT (fixed-parameter tractable) algorithm. The same holds true for the combined parameters $k+d$, $k+\ell$, $\ell+d$ and $k+d+\ell$.
We contrast this extremely negative result with algorithms that manage to escape its hardness.

We show that if we parameterize jointly by the number of agents, the maximum degree of the input graph, and the communication range, then the size of the configuration network is bounded by a function of these parameters (Theorem~\ref{thm:fpt-d-k-Delta}).
Therefore, the problem is in \FPT{} and so is the size of the configuration network that is often used in the A*-based approaches to \MAPFShort{}.
Next, we show that if the input graph~$G$ is a tree, we can obtain an \FPT{} algorithm with respect to the number of agents plus the communication range (Theorem~\ref{thm:fpt-d-k-tree}).
Intuitively, the idea is to use the communication range to prune the input tree and invoke Theorem~\ref{thm:fpt-d-k-Delta}.

It is natural to try to leverage the algorithms from trees to graph families that have bounded treewidth.
We do this for a slightly different combination of parameters.
That is, \MAPFShort{} is \FPT{} for the combination of the treewidth of the graph~$G$, makespan, and the communication range (Theorem~\ref{thm:fpt-tw-d-l}).
This result is achieved by bounding the treewidth of the so-called augmented graph---introduced by \citet{GOR21}---which adds edges between the start and end vertices of each agent.
We then formulate \MAPFShort{} using a Monadic Second Order (MSO) logic formula which can be decided by the result of \citet{Courcelle90}.

This result not only highlights the structure of the augmented graph, which could be of independent interest in future research, but also suggests the potential utility of generic MSO solvers (e.g., \cite{Langer13,BannachB19,Hecher23}) for practical applications.
Moreover, by parameterizing with the number of agents, we extend our results to scenarios based on the local treewidth rather than the global treewidth (Corollary~\ref{cor:fpt-ltw-d-l-k}).
Despite this being a rather technical parameter, it does lead to pertinent results. For example, we obtain an \FPT{} algorithm for planar graphs with respect to the number of agents, makespan, and the communication range (Corollary~\ref{cor:fpt-planar-d-l-k}).
We stress here that many minor-closed graph classes have a bounded local treewidth; the class of planar graphs is just a single representative.
Note that this is in contrast to Theorem~\ref{thm:hard-k-l-d} as the parameterization is the same and the only difference is the graph class.

\section{Preliminaries}
Formally, the input of the \MAPF{} problem consists of a graph $G=(V,E)$, a set of agents~$A$, two functions $s_0\colon A \rightarrow V$, $t\colon A \rightarrow V$ and a positive integer~$\ell$, known as the makespan. The functions $s_0$ and $t$ serve as the starting and target vertices of the agents, respectively. 
For any pair $a , b \in A$ where $a \neq b$, we have that $s_0(a) \neq s_0(b)$ and $t(a) \neq t(b)$.
Initially, each agent $a \in A$ is placed on the vertex $s_0(a)$.
The schedule $s_0, s_1, \ldots, s_\scheduleLength$ assigns each agent a vertex in the given turn $i \in [\scheduleLength]$.
In a specific turn, agents are allowed to move to a vertex neighboring their position in the previous turn, but are not obliged to do so.
The agents can make at most one move per turn, and each vertex can host at most one agent at a given turn.
The position of the agents at the end of the turn $i$ (after the agents have moved) is given by an injective function $s_i\colon A \rightarrow V$.
It is worth mentioning that there are two main variants of the classical \MAPF{} problem, according to whether \textit{swaps} are allowed or not. A swap between two agents $a$ and $b$ during a turn $i$ is defined as the behavior where $s_{i+1}(a)=s_{i}(b)$ and $s_{i+1}(b)=s_{i}(a)$, with $s_i(a)$ and $s_i(b)$ being adjacent. In other words, a swap happens when two agents start from adjacent positions and exchange them within one turn. 

In this paper, we consider the \MAPFCC (\MAPFCCShort for short) problem.
In this generalization of the classical \MAPF problem, each agent has the ability to communicate with other agents that are located within his \emph{communication range}, and it must always be ensured that there is a subset of agents that form a backbone ensuring the communication between all pairs of agents. This communication range is modeled by an integer $d$ that is additionally part of the input.

In order to define what is a feasible solution for the connected variant, we first need to define an auxiliary graph $D$; let us call this the \emph{communication} graph.
First, we set $V(D) = V(G)$.
Then, for every pair $u,v \in V(D)$ we add an edge in $D$ if and only if $\dist_G(u,v)\leq d$.
We say that a vertex set $W \subseteq V(G)$ is \emph{$d$-connected} if the induced subgraph $D[W]$ is connected.
We say that a sequence $s_1,\ldots,s_{\scheduleLength}$ is a \emph{feasible solution} of $\langle G, A, s_0, t,d, \ell\rangle$ if:
\begin{enumerate}
 \item\label{itm:schedule:neighbor} $s_{i}(a)$ is a neighbor of $s_{i-1}(a)$ in $G$, for every agent $a \in A$, $i\in[\scheduleLength]$,
 \item\label{itm:schedule:distinctPositions} for all $i \in [\scheduleLength]$ and $a , b \in A$ where $a \neq b$, we have that $s_i(a) \neq s_i(b)$,
 \item\label{itm:schedule:d-connected} the vertex set $\{s_i(a)\mid a \in A\}$ is $d$-connected, for every $i \in [\scheduleLength]$, and
 \item\label{itm:schedule:targets} for every agent $a\in A$,  we have $s_{\scheduleLength}(a)=t(a)$.
\end{enumerate}

Moreover, we do not allow swaps.
For ease of exposition, we will refer to the condition (\ref{itm:schedule:d-connected}) above as the \emph{communication constraint}.
A feasible solution $s_1,\ldots,s_{\scheduleLength}$ has \emph{makespan}~$\scheduleLength$.
Our goal is to decide if there exists a feasible solution of makespan~$\scheduleLength \leq \ell$. Fig.~\ref{fig:example} illustrates an example of a feasible solution.

\begin{figure}[!t]

\centering

\subfloat[The initial configuration.]{
\begin{tikzpicture}[scale=0.45, inner sep=0.6mm]
    \foreach \x/\y/\color [count=\i] in {0/0/white, 3/0/white, 6/0/white, 9/0/black, 0/3/white, 3/3/white, 6/3/white, 9/3/brown, 0/6/white, 3/6/white, 6/6/white, 9/6/blue, 0/9/white, 3/9/white, 6/9/white, 9/9/red} {
      \node[draw, circle, line width=1pt, fill=\color] (v\i) at (\x,\y) {};
    }
    
    \foreach \x/\y in {1/2, 2/3, 3/4, 5/6, 6/7, 7/8, 9/10, 10/11, 11/12, 13/14, 14/15, 15/16, 1/5, 5/9, 9/13, 4/8, 8/12, 12/16} {
      \draw[line width=0.5pt] (v\x) to (v\y);
    }

    \node[above = 0.1 of v1.west] () {\includegraphics[scale=0.07]{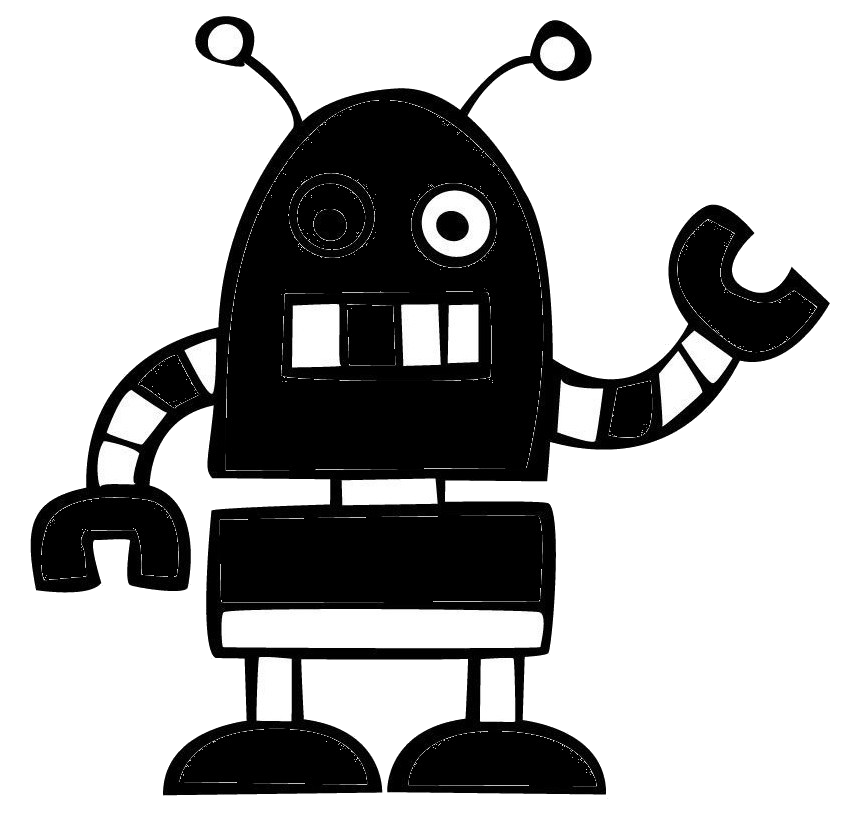}};
    \node[above = 0.1 of v5.west] () {\includegraphics[scale=0.07]{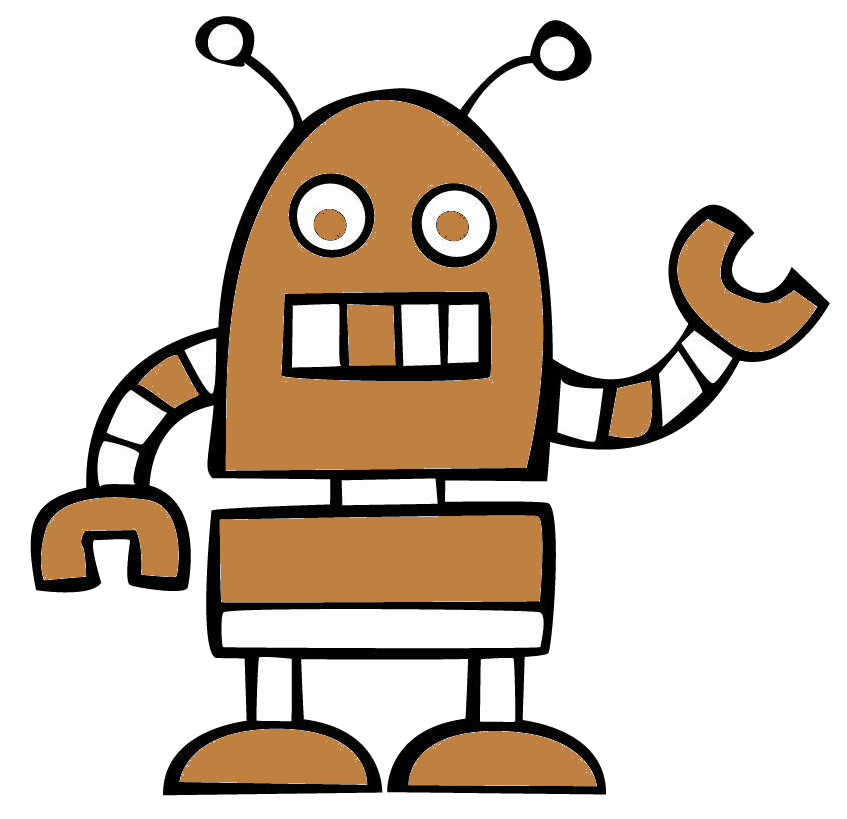}};
    \node[above = 0.1 of v9.west] () {\includegraphics[scale=0.07]{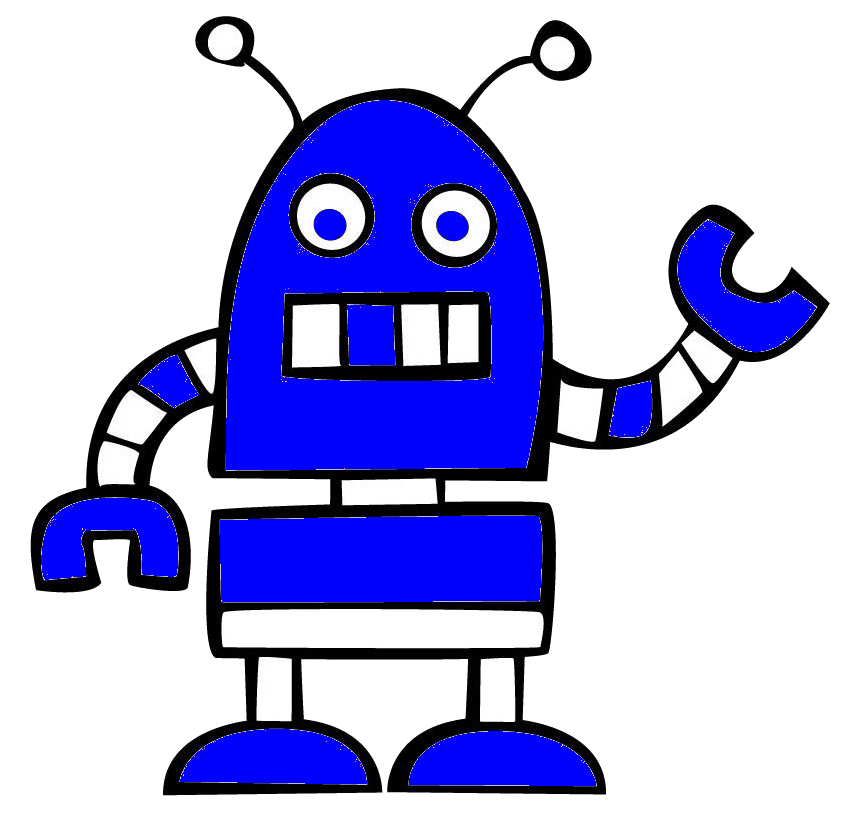}};
    \node[above = 0.1 of v13.west] () {\includegraphics[scale=0.07]{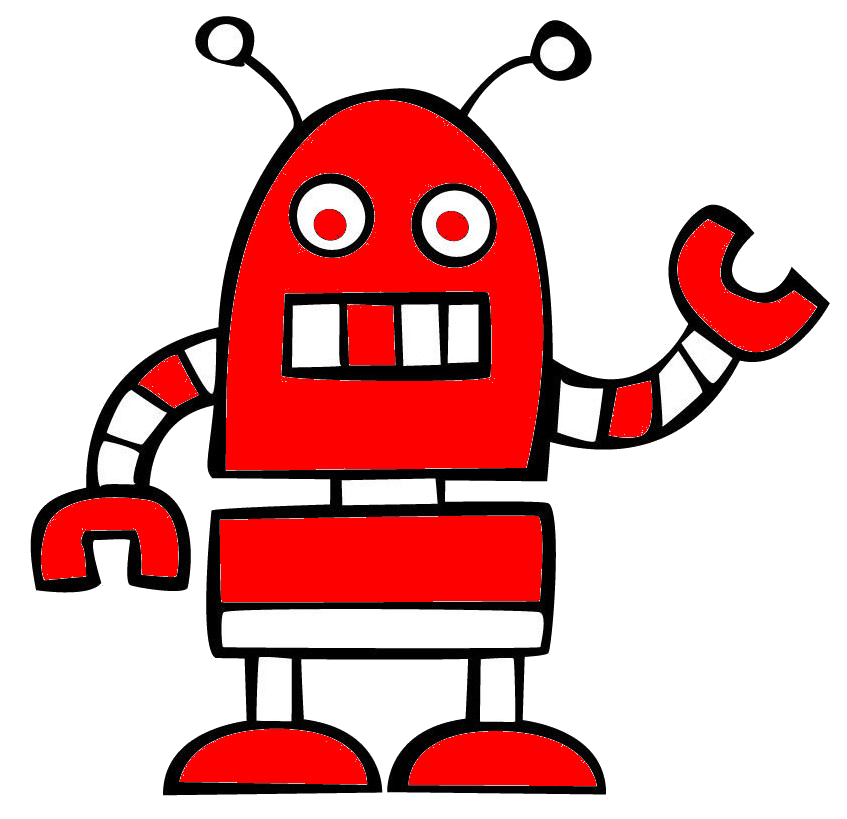}};
    
\end{tikzpicture}
}\hspace{10pt}
\subfloat[Turn $1$.]{
\begin{tikzpicture}[scale=0.45, inner sep=0.6mm]
    \foreach \x/\y/\color [count=\i] in {0/0/white, 3/0/white, 6/0/white, 9/0/black, 0/3/white, 3/3/white, 6/3/white, 9/3/brown, 0/6/white, 3/6/white, 6/6/white, 9/6/blue, 0/9/white, 3/9/white, 6/9/white, 9/9/red} {
      \node[draw, circle, line width=1pt, fill=\color] (v\i) at (\x,\y) {};
    }

    \foreach \x/\y in {1/2, 2/3, 3/4, 5/6, 6/7, 7/8, 9/10, 10/11, 11/12, 13/14, 14/15, 15/16, 1/5, 5/9, 9/13, 4/8, 8/12, 12/16} {
      \draw[line width=0.5pt] (v\x) to (v\y);
    }

    \node[above = 0.1 of v5.west] () {\includegraphics[scale=0.07]{robot-black.png}};
    \node[above = 0.1 of v9.west] () {\includegraphics[scale=0.07]{robot-brown.png}};
    \node[above = 0.01 of v10] () {\includegraphics[scale=0.07]{robot-blue.png}};
    \node[above = 0.1 of v13.west] () {\includegraphics[scale=0.07]{robot-red.png}};
    
\end{tikzpicture}
}\hspace{10pt}
\subfloat[Turn $3$.]{
\begin{tikzpicture}[scale=0.45, inner sep=0.7mm]
    \foreach \x/\y/\color [count=\i] in {0/0/white, 3/0/white, 6/0/white, 9/0/black, 0/3/white, 3/3/white, 6/3/white, 9/3/brown, 0/6/white, 3/6/white, 6/6/white, 9/6/blue, 0/9/white, 3/9/white, 6/9/white, 9/9/red} {
      \node[draw, circle, line width=1pt, fill=\color] (v\i) at (\x,\y) {};
    }
    
    \foreach \x/\y in {1/2, 2/3, 3/4, 5/6, 6/7, 7/8, 9/10, 10/11, 11/12, 13/14, 14/15, 15/16, 1/5, 5/9, 9/13, 4/8, 8/12, 12/16} {
      \draw[line width=0.5pt] (v\x) to (v\y);
    }

    \node[above = 0.01 of v10] () {\includegraphics[scale=0.07]{robot-black.png}};
    \node[above = 0.01 of v11] () {\includegraphics[scale=0.07]{robot-brown.png}};
    \node[above = 0.1 of v12.east] () {\includegraphics[scale=0.07]{robot-blue.png}};
    \node[above = 0.1 of v9.west] () {\includegraphics[scale=0.07]{robot-red.png}};
    
\end{tikzpicture}
}

\subfloat[Turn $6$.]{
\begin{tikzpicture}[scale=0.45, inner sep=0.6mm]
    \foreach \x/\y/\color [count=\i] in {0/0/white, 3/0/white, 6/0/white, 9/0/black, 0/3/white, 3/3/white, 6/3/white, 9/3/brown, 0/6/white, 3/6/white, 6/6/white, 9/6/blue, 0/9/white, 3/9/white, 6/9/white, 9/9/red} {
      \node[draw, circle, line width=1pt, fill=\color] (v\i) at (\x,\y) {};
    }

    \foreach \x/\y in {1/2, 2/3, 3/4, 5/6, 6/7, 7/8, 9/10, 10/11, 11/12, 13/14, 14/15, 15/16, 1/5, 5/9, 9/13, 4/8, 8/12, 12/16} {
      \draw[line width=0.5pt] (v\x) to (v\y);
    }

    \node[above = 0.1 of v12.east] () {\includegraphics[scale=0.07]{robot-red.png}};
    \node[above = 0.01 of v7] () {\includegraphics[scale=0.07]{robot-brown.png}};
    \node[above = 0.1 of v8.east] () {\includegraphics[scale=0.07]{robot-black.png}};
    \node[above = 0.01 of v6] () {\includegraphics[scale=0.07]{robot-blue.png}};
    
\end{tikzpicture}
}\hspace{10pt}
\subfloat[Turn $8$.]{
\begin{tikzpicture}[scale=0.45, inner sep=0.6mm]
    \foreach \x/\y/\color [count=\i] in {0/0/white, 3/0/white, 6/0/white, 9/0/black, 0/3/white, 3/3/white, 6/3/white, 9/3/brown, 0/6/white, 3/6/white, 6/6/white, 9/6/blue, 0/9/white, 3/9/white, 6/9/white, 9/9/red} {
      \node[draw, circle, line width=1pt, fill=\color] (v\i) at (\x,\y) {};
    }

    \foreach \x/\y in {1/2, 2/3, 3/4, 5/6, 6/7, 7/8, 9/10, 10/11, 11/12, 13/14, 14/15, 15/16, 1/5, 5/9, 9/13, 4/8, 8/12, 12/16} {
      \draw[line width=0.5pt] (v\x) to (v\y);
    }

    \node[above = 0.1 of v12.east] () {\includegraphics[scale=0.07]{robot-red.png}};
    \node[above = 0.1 of v4.east] () {\includegraphics[scale=0.07]{robot-brown.png}};
    \node[above = 0.01 of v3] () {\includegraphics[scale=0.07]{robot-black.png}};
    \node[above = 0.1 of v8.east] () {\includegraphics[scale=0.07]{robot-blue.png}};
    
\end{tikzpicture}
}\hspace{10pt}
\subfloat[Turn $9$.]{
\begin{tikzpicture}[scale=0.45, inner sep=0.6mm]
    \foreach \x/\y/\color [count=\i] in {0/0/white, 3/0/white, 6/0/white, 9/0/black, 0/3/white, 3/3/white, 6/3/white, 9/3/brown, 0/6/white, 3/6/white, 6/6/white, 9/6/blue, 0/9/white, 3/9/white, 6/9/white, 9/9/red} {
      \node[draw, circle, line width=1pt, fill=\color] (v\i) at (\x,\y) {};
    }

    \foreach \x/\y in {1/2, 2/3, 3/4, 5/6, 6/7, 7/8, 9/10, 10/11, 11/12, 13/14, 14/15, 15/16, 1/5, 5/9, 9/13, 4/8, 8/12, 12/16} {
      \draw[line width=0.5pt] (v\x) to (v\y);
    }

    \node[above = 0.1 of v16.east] () {\includegraphics[scale=0.07]{robot-red.png}};
    \node[above = 0.1 of v12.east] () {\includegraphics[scale=0.07]{robot-blue.png}};
    \node[above = 0.1 of v8.east] () {\includegraphics[scale=0.07]{robot-brown.png}};
    \node[above = 0.1 of v4.east] () {\includegraphics[scale=0.07]{robot-black.png}};
    
\end{tikzpicture}
}
\caption{An example of a feasible solution for the instance encoded in subfigure (a), for a communication constraint of $d=1$. The colors on the agents and the vertices are used to encode the terminal vertex of each agent. 
Observe also that if each agent followed their optimal path to reach their terminal (the horizontal path on which they start), we would have a schedule of length $2$, but this would not be a feasible solution, as the communication constraint would not be respected during the first and second turns.}\label{fig:example}
\end{figure}

\paragraph*{Parametrized Complexity.}
The parametrized complexity point of view allows us to overcome the limitations of classical measures of time (and space) complexity, by taking into account additional measures that can affect the running time of an algorithm; these additional measures are exactly what we refer to as parameters. The goal here is to construct exact algorithms that run in time $f(k)\cdot\operatorname{poly}(n)$, where $f$ is a computable function, $n$ is the size of the input and $k$ is the parameter. Algorithms with such running times are referred to as \textit{fixed-parameter tractable} (FPT).
A problem admitting such an algorithm belongs to the class \FPT. 
Similar to classical complexity theory, there is also a notion of infeasibility. A problem is presumably not in FPT{} if it is shown to be \W[1]-hard (by a parameterized reduction).
We refer the interested reader to now classical monographs~\cite{CyganFKLMPPS15,Niedermeier06,FlumG06,downey2012parameterized} for a more comprehensive introduction to this topic.

\paragraph*{Structural Parameters and Logic.}
The main structural parameter that interests us in this work is that of \textit{treewidth}.
%
A \emph{tree-decomposition} of $G$ is a pair $(\mathcal{T},\beta)$, where~$\mathcal{T}$ is a tree, $\beta\colon V(\mathcal{T})\to 2^{V}$ is a function assigning each node $x$ of $\mathcal{T}$ its \emph{bag}, and the following conditions hold:
\begin{itemize}
	\item for every edge $\{u,v\}\in E(G)$ there is a node $x\in V(\mathcal{T})$ such that $u,v\in\beta(x)$, and
	\item for every vertex $v\in V$, the set of nodes $x$ with $v\in\beta(x)$ induces a connected subtree of $\mathcal{T}$.
\end{itemize}
The \emph{width} of a tree-decomposition $(\mathcal{T},\beta)$ is $\max_{x\in V(\mathcal{T})} |\beta(x)|-1$, and the treewidth $\tw(G)$ of a graph $G$ is the minimum width of a tree-decomposition of~$G$. In Figure~\ref{fig-tree-deco} we illustrate an example of the above notions.

It is known that computing a tree-decomposition of minimum width \NP-hard~\cite{ACP87}, but it is fixed-parameter tractable when parameterized by the treewidth~\cite{Kloks94,Bodlaender96}, and even more efficient algorithms exist for obtaining near-optimal tree-decompositions~\cite{KorhonenL23}.

\begin{figure}[!t]
\centering
\subfloat[The graph $G$]{
\begin{tikzpicture}[scale=0.65, inner sep=0.7mm]
   
    
    \node[draw, circle, line width=0.5pt, fill=white](u1) at (1.5,-1.5)[label=below: $u_1$] {};
    \node[draw, circle, line width=0.5pt, fill=white](u2) at (0,0)[label=left: $u_2$] {};
    \node[draw, circle, line width=0.5pt, fill=white](u3) at (3,0)[label=right: $u_3$] {};
    \node[draw, circle, line width=0.5pt, fill=white](u4) at (1.5,1.5)[label=left: $u_4$] {};

    \node[draw, circle, line width=0.5pt, fill=white](u5) at (1.5,3)[label=left: $u_5$] {};
    \node[draw, circle, line width=0.5pt, fill=white](u6) at (0,4.5)[label=left: $u_6$] {};
    \node[draw, circle, line width=0.5pt, fill=blue](u7) at (3,4.5)[label=above: $\textcolor{blue}{u_7}$] {};
    \node[draw, circle, line width=0.5pt, fill=white](u8) at (1.5,6)[label=above: $u_8$] {};

    \node[draw, circle, line width=0.5pt, fill=white](u9) at (4.5,6)[label=above: $u_9$] {};
    \node[draw, circle, line width=0.5pt, fill=white](u10) at (6,4.5)[label=right: $u_{10}$] {};
    \node[draw, circle, line width=0.5pt, fill=white](u11) at (6,3)[label=right: $u_{11}$] {};

    \draw[-, line width=0.5pt]  (u1) -- (u2);
    \draw[-, line width=0.5pt]  (u2) -- (u4);
    \draw[-, line width=0.5pt]  (u3) -- (u4);
    \draw[-, line width=0.5pt]  (u1) -- (u4);

    \draw[-, line width=0.5pt]  (u4) -- (u5);
    \draw[-, line width=0.5pt]  (u4) -- (u7);
    \draw[-, line width=0.5pt]  (u5) -- (u6);
    \draw[-, line width=0.5pt]  (u5) -- (u7);
    \draw[-, line width=0.5pt]  (u7) -- (u8);
    \draw[-, line width=0.5pt]  (u6) -- (u7);

    \draw[-, line width=0.5pt]  (u7) -- (u9);
    \draw[-, line width=0.5pt]  (u7) -- (u10);
    \draw[-, line width=0.5pt]  (u9) -- (u10);
    \draw[-, line width=0.5pt]  (u10) -- (u11);

\end{tikzpicture}}\hspace{20pt}
\subfloat[Assigning the vertices of $G$ to bags of size at most $3$]{
\begin{tikzpicture}[scale=0.65, inner sep=0.7mm]
   
    
    \node[draw, circle, line width=0.5pt, fill=white](u1) at (1.5,-1.5)[] {};
    \node[draw, circle, line width=0.5pt, fill=white](u2) at (0,0)[] {};
    \node[draw, circle, line width=0.5pt, fill=white](u3) at (3,0)[] {};
    \node[draw, circle, line width=0.5pt, fill=white](u4) at (1.5,1.5)[] {};

    \node[draw, circle, line width=0.5pt, fill=white](u5) at (1.5,3)[] {};
    \node[draw, circle, line width=0.5pt, fill=white](u6) at (0,4.5)[] {};
    \node[draw, circle, line width=0.5pt, fill=blue](u7) at (3,4.5)[] {};
    \node[draw, circle, line width=0.5pt, fill=white](u8) at (1.5,6)[] {};

    \node[draw, circle, line width=0.5pt, fill=white](u9) at (4.5,6)[] {};
    \node[draw, circle, line width=0.5pt, fill=white](u10) at (6,4.5)[] {};
    \node[draw, circle, line width=0.5pt, fill=white](u11) at (6,3)[] {};

    \draw[-, line width=0.5pt]  (u1) -- (u2);
    \draw[-, line width=0.5pt]  (u2) -- (u4);
    \draw[-, line width=0.5pt]  (u3) -- (u4);
    \draw[-, line width=0.5pt]  (u1) -- (u4);

    \draw[-, line width=0.5pt]  (u4) -- (u5);
    \draw[-, line width=0.5pt]  (u4) -- (u7);
    \draw[-, line width=0.5pt]  (u5) -- (u6);
    \draw[-, line width=0.5pt]  (u5) -- (u7);
    \draw[-, line width=0.5pt]  (u7) -- (u8);
    \draw[-, line width=0.5pt]  (u6) -- (u7);

    \draw[-, line width=0.5pt]  (u7) -- (u9);
    \draw[-, line width=0.5pt]  (u7) -- (u10);
    \draw[-, line width=0.5pt]  (u9) -- (u10);
    \draw[-, line width=0.5pt]  (u10) -- (u11);
\draw[dashed, line width=1pt, color=brown, rounded corners=6pt]
    ($(u4)+(-0.5,0.5)$) --
    ($(u4)+(0.5,0.5)$) --
    ($(u1)+(0.5,-0.5)$) --
    ($(u1)+(-0.5,-0.5)$) --
    ($(u2)+(-0.5,-0.5)$) --
    ($(u2)+(-0.5,0.5)$) -- cycle;
\draw[dashed, line width=1pt, color=purple, rounded corners=6pt]
    ($(u3)+(0.5,0.5)$) --
    ($(u3)+(0.5,-0.5)$) --
    ($(u3)+(-0.5,-0.5)$) --
    ($(u4)+(-0.5,-0.5)$) --  
    ($(u4)+(-0.5,0.5)$) -- 
    ($(u4)+(0.5,0.5)$) -- cycle;
\draw[dashed, line width=1pt, rounded corners=6pt]
    ($(u4)+(-0.5,-0.5)$) --
    ($(u4)+(0.5,-0.5)$) --
    ($(u7)+(0.5,-0.5)$) --
    ($(u7)+(0.5,0.5)$) --
    ($(u7)+(-0.5,0.5)$) --
    ($(u5)+(-0.5,0.5)$) -- cycle;
\draw[dashed, line width=1pt, color=orange, rounded corners=6pt]
    ($(u11)+(-0.5,-0.5)$) --
    ($(u11)+(0.5,-0.5)$) --
    ($(u10)+(0.5,0.5)$) -- 
    ($(u10)+(-0.5,0.5)$) -- cycle;
\draw[dashed, line width=1pt, color=green!60!black, rounded corners=6pt]
    ($(u10)+(0.5,0.5)$) --
    ($(u10)+(0.5,-0.5)$) --
    ($(u7)+(-0.5,-0.5)$) --
    ($(u7)+(-0.5,0.5)$) --
    ($(u9)+(-0.5,0.5)$) -- 
    ($(u9)+(0.5,0.5)$) -- cycle;
\draw[dashed, line width=1pt, color=blue, rounded corners=6pt]
    ($(u7)+(0.5,0.5)$) --
    ($(u7)+(0.5,-0.5)$) --
    ($(u7)+(-0.5,-0.5)$) --
    ($(u8)+(-0.5,-0.5)$) --  
    ($(u8)+(-0.5,0.5)$) -- 
    ($(u8)+(0.5,0.5)$) -- cycle;
    
\draw[dashed, line width=1pt, color=red, rounded corners=6pt]
    ($(u7)+(0.5,-0.5)$) --
    ($(u7)+(0.5,0.5)$) --
    ($(u6)+(-0.5,0.5)$) --
    ($(u6)+(-0.5,-0.5)$) --
    ($(u5)+(-0.5,-0.5)$) -- 
    ($(u5)+(0.5,-0.5)$) -- cycle;
    
\end{tikzpicture}}\hspace{20pt}
\subfloat[A tree-decomposition $\mathcal{T}$ of $G$ of width $2$]{
\begin{tikzpicture}[scale=0.55, inner sep=0.7mm]
   
    
    \node[draw, circle, color=brown, line width=1pt, fill=white](b1) at (0,0)[] {\textcolor{black}{$u_1,u_2,u_4$}};
    \node[draw, circle, color=purple, line width=1pt, fill=white](b7) at (4,0)[] {\textcolor{black}{$u_3,u_4$}};
    \node[draw, circle, line width=1pt, fill=white](b2) at (0,3)[] {$\textcolor{black}{u_4,u_5,}\textcolor{blue}{u_7}$};
    \node[draw, circle, color=red, line width=1pt, fill=white](b3) at (0,6)[] {$\textcolor{black}{u_5,u_6,}\textcolor{blue}{u_7}$};
    \node[draw, circle, color=blue, line width=1pt, fill=white](b4) at (0,9)[] {$\textcolor{blue}{u_7}\textcolor{black}{,u_8}$};

    \node[draw, circle, color=green!60!black, line width=1pt, fill=white](b5) at (4,6)[] {$\textcolor{blue}{u_7},u_9,u_{10}$};
    \node[draw, circle, color=orange, line width=1pt, fill=white](b6) at (4,3)[] {\textcolor{black}{$u_{10},u_{11}$}};

    \draw[-, line width=0.5pt]  (b1) -- (b2);
    \draw[-, line width=0.5pt]  (b2) -- (b7);
    \draw[-, line width=0.5pt]  (b2) -- (b3);
    \draw[-, line width=0.5pt]  (b3) -- (b4);
    \draw[-, line width=0.5pt]  (b3) -- (b5);
    \draw[-, line width=0.5pt]  (b5) -- (b6);

\end{tikzpicture}
}
\caption{Illustration of a tree-decomposition. Every edge of $G$ appears in at least one node of $\mathcal{T}$. For every vertex of $G$ (e.g. $u_7$), the nodes that contain that vertex form a tree of $\mathcal{T}$. This decomposition is of width $2$ as its largest bag is of order $3$.}\label{fig-tree-deco}
\end{figure}

In our work, we make use of the celebrated Courcelle Theorem~\cite{Courcelle90}, stating that any problem that is expressible by a monadic second-order formula can be solved in FPT-time parameterized by the treewidth of~$G$.
\textsf{MSO} logic is an extension of first-order logic, distinguished by the introduction of set variables (denoted by uppercase letters) that represent sets of domain elements, in contrast to individual variables (denoted by lowercase letters), which represent single elements.
Specifically, we utilize \MSOtwo, a variant of \textsf{MSO} logic that allows quantification over both the vertices and edges.
This extension enables us to address a broader class of problems.
More generally, Courcelle's algorithm extends to the case when both the graph~$G$ and the \MSOtwo language are enriched with finitely many vertex and edge labels.

\section{The Problem is Very Hard}
In this section, we will prove the following theorem. 
\begin{theorem}\label{thm:hard-k-l-d}
 The \MAPFCC{} problem is \W[1]-hard parameterized by the number of agents, even for $\ell=3$ and $d=1$.
\end{theorem}
\begin{proof}
The reduction is from the \textsc{$k$-Multicolored Clique} ($k$-MCC for short) problem. This problem takes as input a graph $H=(V,E)$, whose vertex set $V$ is partitioned into the $k$ independent sets $S_1,\dots,S_k$. The question is whether there exists a clique on $k$ vertices as a subgraph of $H$. Observe that if such a clique does exist, then it contains a unique vertex from $S_i$, for each $i\in [k]$. This problem was shown to be \W[1]-hard in~\cite{Pietrzak03}.

Starting from an input of the $k$-MCC problem, consisting of a graph $H$ whose vertex set is partitioned into sets $S_1,\dots,S_k$, we will construct an instance $I=\langle G,A,s_0,t,1,3\rangle$ of \MAPFCCShort such that $I$ is a yes-instance if and only if $H$ contains a clique on $k$ vertices.

\paragraph*{The construction of $G$.} First, we describe the two gadgets that will serve as the building blocks of $G$. For each $i\in[k]$, let $S_i=\{v^i_1,\dots,v^i_n\}$; we build the $V_i$ gadget as follows (see Fig.~\ref{fig:gadget1-hard-k-l-d}). We begin with $n$ paths with $k-1$ vertices each, each corresponding to a vertex of $S_i$. So, for each $p\in[n]$, we have the path $P_p^i=v_{p,1}^iv_{p,2}^i\dots v_{p,i-1}^iv_{p,i+1}^i\dots v_{p,k}^i$, which excludes the vertex $v_{p,i}^i$. We then add the new path $a^i_1\dots a^i_{i-1}a^i_{i+1}\dots a^i_k$ and, for each $p\in [n]$, we add the edge $a^i_jv_{p,j}^i$, for each $j\in[k]\setminus i$. We say that the vertices $v_{p,1}^i$ and $v_{p,k}^i$, for every $p\in[n]$ are the \emph{top} and \emph{bottom} vertices of the $V_i$ gadget, respectively (if $i=1$ or $i=k$ we adapt accordingly). Next, for each $l,m\in [k]$ with $l<m$, we build the $E_{l,m} $ gadget.
This gadget consists of a forest of edges, each corresponding to an edge between the vertices of $S_l$ and $S_m$ in $H$. That is, there exist vertices $u_p^{l,m}$ and $u_q^{m,l}$ in $E_{l,m}$ such that $u_p^{l,m}u_q^{m,l}\in E(E_{l,m})$ if and only if there exist vertices $v_p^l\in S_l$ and $v_q^m\in S_m$ such that $v_p^lv_q^m\in E(H)$.
We say that $u_p^{l,m}$ and $u_q^{m,l}$, for every $p,q\in[n]$, are the \emph{top} and \emph{bottom}, respectively, vertices of $E_{l,m}$. Note that $E_{l,m}$ contains at most $O(n^2)$ vertices. Moreover, since $l<m$, we have $\binom{k}{2}$ such gadgets in total. 
This finishes the construction of the two gadgets we use.

\begin{figure}[!t]
\centering

\begin{tikzpicture}[scale=0.6, inner sep=0.7mm]
    \node[draw, circle, line width=1pt, fill=white](a1) at (0,12)[label=above left: $a_1^i$] {};
    \node[draw, circle, line width=1pt, fill=white](a2) at (0,9)[label=above left: $a_2^i$] {};
    \node[draw, circle, line width=1pt, fill=white](a3) at (0,6)[label=above left: $a_{i-1}^i$] {};
    \node[draw, circle, line width=1pt, fill=white](a4) at (0,3)[label=above left: $a_{i+1}^i$] {};
    \node[draw, circle, line width=1pt, fill=white](a5) at (0,0)[label=above left: $a_k^i$] {};

    \node[draw, circle, line width=1pt, fill=gray!50](v11) at (3,12)[label=right: $v_{1,1}^i$] {};
    \node[draw, circle, line width=1pt, fill=white](v21) at (3,9)[label=right: $v_{1,2}^i$] {};
    \node[draw, circle, line width=1pt, fill=white](v31) at (3,6)[label=right: $v_{1,i-1}^i$] {};
    \node[draw, circle, line width=1pt, fill=white](v41) at (3,3)[label=right: $v_{1,i+1}^i$] {};
    \node[draw, circle, line width=1pt, fill=black](v51) at (3,0)[label=right: $v_{1,k}^i$] {};

    \node[draw, circle, line width=1pt, fill=gray!50](v12) at (6,12)[label=right: $v_{2,1}^i$] {};
    \node[draw, circle, line width=1pt, fill=white](v22) at (6,9)[label=right: $v_{2,2}^i$] {};
    \node[draw, circle, line width=1pt, fill=white](v32) at (6,6)[label=right: $v_{2,i-1}^i$] {};
    \node[draw, circle, line width=1pt, fill=white](v42) at (6,3)[label=right: $v_{2,i+1}^i$] {};
    \node[draw, circle, line width=1pt, fill=black](v52) at (6,0)[label=right: $v_{2,k}^i$] {};

    \node[draw, circle, line width=1pt, fill=gray!50](v13) at (11,12)[label=right: $v_{n,1}^i$] {};
    \node[draw, circle, line width=1pt, fill=white](v23) at (11,9)[label=right: $v_{n,2}^i$] {};
    \node[draw, circle, line width=1pt, fill=white](v33) at (11,6)[label=right: $v_{n,i-1}^i$] {};
    \node[draw, circle, line width=1pt, fill=white](v43) at (11,3)[label=right: $v_{n,i+1}^i$] {};
    \node[draw, circle, line width=1pt, fill=black](v53) at (11,0)[label=right: $v_{n,k}^i$] {};

    \draw[middle dotted line] (v21) -- (v31);
    \draw[middle dotted line] (v41) -- (v51);
    \draw[middle dotted line] (v22) -- (v32);
    \draw[middle dotted line] (v42) -- (v52);
    \draw[middle dotted line] (v23) -- (v33);
    \draw[middle dotted line] (v43) -- (v53);

    \path (v32) -- (v33) node [black, midway, sloped] {$\dots$};
    \draw[middle dotted line] (a2) -- (a3);
    \draw[middle dotted line] (a4) -- (a5);

	\draw[-, line width=0.5pt]  (a1) -- (v11);
	\draw[-, line width=0.5pt]  (a2) -- (v21);
	\draw[-, line width=0.5pt]  (a3) -- (v31);
	\draw[-, line width=0.5pt]  (a4) -- (v41);
	\draw[-, line width=0.5pt]  (a5) -- (v51);

	\draw[-, line width=0.5pt]  (v11) -- (v21);
	\draw[-, line width=0.5pt]  (v31) -- (v41);
	\draw[-, line width=0.5pt]  (v12) -- (v22);
	\draw[-, line width=0.5pt]  (v32) -- (v42);
	\draw[-, line width=0.5pt]  (v13) -- (v23);
	\draw[-, line width=0.5pt]  (v33) -- (v43);

	\draw[-, line width=0.5pt]  (a1) -- (a2);
	\draw[-, line width=0.5pt]  (a3) -- (a4);

	\draw (a1) edge[bend left=20] (v12);
	\draw (a1) edge[bend left=20] (v13);
	\draw (a2) edge[bend left=20] (v22);
	\draw (a2) edge[bend left=20] (v23);
	\draw (a3) edge[bend left=20] (v32);
	\draw (a3) edge[bend left=20] (v33);
	\draw (a4) edge[bend left=20] (v42);
	\draw (a4) edge[bend left=20] (v43);
	\draw (a5) edge[bend left=20] (v52);
	\draw (a5) edge[bend left=20] (v53);
\end{tikzpicture}
\caption{The $V_i$ gadget, for any $i\in [k]$, used in the proof of Theorem~\ref{thm:hard-k-l-d}. The color gray (black resp.) is used to represent the top (bottom resp.) vertices of the gadget.}\label{fig:gadget1-hard-k-l-d}
\end{figure}
We are now ready to construct the graph $G$. We start with a copy of the $V_i$ gadget for each $i\in[k]$ and a copy of the $E_{l,m}$ gadget for each $l,m\in[k]$ with $l<m$. For each $i\in [k-1]$, we add all the edges between the bottom vertices of $V_{i}$ and the top vertices of $V_{i+1}$, as well as the edge $a_k^ia_1^{i+1}$. Then, we connect the $V_i$ gadgets with the $E_{l,m}$ gadgets as follows (illustrated in Fig.~\ref{fig:connection-gadgets-hard-k-l-d}). For every $l,m\in[k]$ with $l<m$ and $p,q\in[n]$, for every $u^{l,m}_p\in E_{l,m}$ and $u^{m,l}_q\in E_{l,m}$, we connect $u^{l,m}_p$ with $v_{p,m}^l$ and $u^{m,l}_q$ with $v_{q,l}^m$. Next, for every $l,m\in[k-1]$ with $l<m$, we add all the edges between the bottom vertices of $E_{l,m}$ and the top vertices of $E_{l,m+1}$ as well as all the edges between the bottom vertices of $E_{l,k}$ and the top vertices of $E_{l+1,l+2}$.
Then we add the clique $Q$ with the $k(k-1)$ vertices $\{t_j^i : i\in [k] \text{ and }  j\in [k]\setminus \{i\}\}$, and we connect all the vertices of every $E_{l,m}$ gadget to all the vertices of~$Q$.
\begin{figure}[!t]
\centering

\begin{tikzpicture}[scale=0.4, inner sep=0.4mm]

    \node[draw, circle, line width=1pt, fill=white](a11) at (0,0)[] {};
    \node[draw, circle, line width=1pt, fill=white](v11) at (3,0)[] {};
    \node[draw, circle, line width=1pt, fill=white](v12) at (6,0)[] {};
    \node[draw, circle, line width=1pt, fill=white](v13) at (9,0)[] {};
    \node[draw, circle, line width=1pt, fill=white](v14) at (12,0)[] {};
    \node[draw, circle, line width=1pt, fill=white](a12) at (0,-2)[] {};
    \node[draw, circle, line width=1pt, fill=white](v15) at (3,-2)[] {};
    \node[draw, circle, line width=1pt, fill=white](v16) at (6,-2)[] {};
    \node[draw, circle, line width=1pt, fill=white](v17) at (9,-2)[] {};
    \node[draw, circle, line width=1pt, fill=white](v18) at (12,-2)[] {};
    \node[draw, circle, line width=1pt, fill=white](a13) at (0,-6)[] {};
    \node[draw, circle, line width=1pt, fill=white](v19) at (3,-6)[] {};
    \node[draw, circle, line width=1pt, fill=white](v110) at (6,-6)[label=left: $v^1_{2,k}$] {};
    \node[draw, circle, line width=1pt, fill=white](v111) at (9,-6)[] {};
    \node[draw, circle, line width=1pt, fill=white](v112) at (12,-6)[] {};

    \path (v13) -- (v14) node [black, midway, sloped] {$\dots$};
    \path (v17) -- (v18) node [black, midway, sloped] {$\dots$};
    \path (v111) -- (v112) node [black, midway, sloped] {$\dots$};
    \draw[middle dotted line] (a12) -- (a13);
    \draw[middle dotted line] (v15) -- (v19);
    \draw[middle dotted line] (v16) -- (v110);
    \draw[middle dotted line] (v17) -- (v111);
    \draw[middle dotted line] (v18) -- (v112);

    \draw[-, line width=0.5pt]  (v11) -- (v15);
    \draw[-, line width=0.5pt]  (v12) -- (v16);
    \draw[-, line width=0.5pt]  (v13) -- (v17);
    \draw[-, line width=0.5pt]  (v14) -- (v18);
    \draw[-, line width=0.5pt]  (a11) -- (a12);

	\draw[-, line width=0.5pt,gray!30] (a11) edge[bend left=20] (v11);
	\draw[-, line width=0.5pt,gray!30] (a11) edge[bend left=20] (v12);
	\draw[-, line width=0.5pt,gray!30] (a11) edge[bend left=20] (v11);
	\draw[-, line width=0.5pt,gray!30] (a11) edge[bend left=20] (v13);
	\draw[-, line width=0.5pt,gray!30] (a11) edge[bend left=20] (v14);

	\draw[-, line width=0.5pt,gray!30] (a12) edge[bend left=20] (v15);
	\draw[-, line width=0.5pt,gray!30] (a12) edge[bend left=20] (v16);
	\draw[-, line width=0.5pt,gray!30] (a12) edge[bend left=20] (v17);
	\draw[-, line width=0.5pt,gray!30] (a12) edge[bend left=20] (v18);

	\draw[-, line width=0.5pt,gray!30] (a13) edge[bend left=20] (v19);
	\draw[-, line width=0.5pt,gray!30] (a13) edge[bend left=20] (v110);
	\draw[-, line width=0.5pt,gray!30] (a13) edge[bend left=20] (v111);
	\draw[-, line width=0.5pt,gray!30] (a13) edge[bend left=20] (v112);
	\node[draw, circle, line width=1pt, fill=white]() at (6,-6)[label=left: $v^1_{2,k}$] {};

    \draw[-, line width=0.5pt,gray!50] (-1,1.5) -- (13,1.5);
    \draw[-, line width=0.5pt,gray!50] (-1,1.5) -- (-1,-7);
    \draw[-, line width=0.5pt,gray!50] (13,1.5) -- (13,-7);
    \draw[-, line width=0.5pt,gray!50] (-1,-7) -- (13,-7);
    \node at (-0.5,2.1) {$V_1$};

    \node[draw, circle, line width=1pt, fill=white](a31) at (0,-11)[] {};
    \node[draw, circle, line width=1pt, fill=white](v31) at (3,-11)[label=left: $v^k_{1,1}$] {};
    \node[draw, circle, line width=1pt, fill=white](v32) at (6,-11)[] {};
    \node[draw, circle, line width=1pt, fill=white](v33) at (9,-11)[] {};
    \node[draw, circle, line width=1pt, fill=white](v34) at (12,-11)[] {};
    \node[draw, circle, line width=1pt, fill=white](a32) at (0,-13)[] {};
    \node[draw, circle, line width=1pt, fill=white](v35) at (3,-13)[] {};
    \node[draw, circle, line width=1pt, fill=white](v36) at (6,-13)[] {};
    \node[draw, circle, line width=1pt, fill=white](v37) at (9,-13)[] {};
    \node[draw, circle, line width=1pt, fill=white](v38) at (12,-13)[] {};
    \node[draw, circle, line width=1pt, fill=white](a33) at (0,-17)[] {};
    \node[draw, circle, line width=1pt, fill=white](v39) at (3,-17)[] {};
    \node[draw, circle, line width=1pt, fill=white](v310) at (6,-17)[] {};
    \node[draw, circle, line width=1pt, fill=white](v311) at (9,-17)[] {};
    \node[draw, circle, line width=1pt, fill=white](v312) at (12,-17)[] {};

    \path (v33) -- (v34) node [black, midway, sloped] {$\dots$};
    \path (v37) -- (v38) node [black, midway, sloped] {$\dots$};
    \path (v311) -- (v312) node [black, midway, sloped] {$\dots$};
    \draw[middle dotted line] (a32) -- (a33);
    \draw[middle dotted line] (v35) -- (v39);
    \draw[middle dotted line] (v36) -- (v310);
    \draw[middle dotted line] (v37) -- (v311);
    \draw[middle dotted line] (v38) -- (v312);

    \draw[-, line width=0.5pt]  (v31) -- (v35);
    \draw[-, line width=0.5pt]  (v32) -- (v36);
    \draw[-, line width=0.5pt]  (v33) -- (v37);
    \draw[-, line width=0.5pt]  (v34) -- (v38);
    \draw[-, line width=0.5pt]  (a31) -- (a32);

	\draw[-, line width=0.5pt,gray!30] (a31) edge[bend left=20] (v31);
	\draw[-, line width=0.5pt,gray!30] (a31) edge[bend left=20] (v32);
	\draw[-, line width=0.5pt,gray!30] (a31) edge[bend left=20] (v31);
	\draw[-, line width=0.5pt,gray!30] (a31) edge[bend left=20] (v33);
	\draw[-, line width=0.5pt,gray!30] (a31) edge[bend left=20] (v34);

	\draw[-, line width=0.5pt,gray!30] (a32) edge[bend left=20] (v35);
	\draw[-, line width=0.5pt,gray!30] (a32) edge[bend left=20] (v36);
	\draw[-, line width=0.5pt,gray!30] (a32) edge[bend left=20] (v37);
	\draw[-, line width=0.5pt,gray!30] (a32) edge[bend left=20] (v38);

	\draw[-, line width=0.5pt,gray!30] (a33) edge[bend left=20] (v39);
	\draw[-, line width=0.5pt,gray!30] (a33) edge[bend left=20] (v310);
	\draw[-, line width=0.5pt,gray!30] (a33) edge[bend left=20] (v311);
	\draw[-, line width=0.5pt,gray!30] (a33) edge[bend left=20] (v312);
	\node[draw, circle, line width=1pt, fill=white]() at (3,-11)[label=left: $v^k_{1,1}$] {};

	\path (6,-6) -- (6,-11) node [black, midway, sloped] {$\dots$};

    \draw[-, line width=0.5pt,gray!50] (-1,-9.5) -- (13,-9.5);
    \draw[-, line width=0.5pt,gray!50] (-1,-9.5) -- (-1,-18);
    \draw[-, line width=0.5pt,gray!50] (13,-9.5) -- (13,-18);
    \draw[-, line width=0.5pt,gray!50] (-1,-18) -- (13,-18);
    \node at (-0.5,-8.9) {$V_k$};

    \node[draw, circle, line width=1pt, fill=white](u1k1) at (18,-7.5)[] {};
    \node[draw, circle, line width=1pt, fill=white](u1k2) at (20,-7.5)[label=right: $u^{1,k}_2$] {};
    \node[draw, circle, line width=1pt, fill=white](u1k3) at (24,-7.5)[] {};
    \node[draw, circle, line width=1pt, fill=white](uk11) at (18,-9.5)[] {};
    \node[draw, circle, line width=1pt, fill=white](uk12) at (20,-9.5)[label=right: $u^{k,1}_1$] {};
    \node[draw, circle, line width=1pt, fill=white](uk13) at (24,-9.5)[] {};

    \draw[-, line width=0.5pt]  (u1k1) -- (uk11);
    \draw[-, line width=0.5pt]  (u1k2) -- (uk12);
    \draw[-, line width=0.5pt]  (u1k3) -- (uk13);

    \path (20,-8.5) -- (24,-8.5) node [black, midway, sloped] {$\dots$};


    \draw[-, line width=0.5pt,gray!50] (17,-6.5) -- (25,-6.5);
    \draw[-, line width=0.5pt,gray!50] (25,-6.5) -- (25,-10.5);
    \draw[-, line width=0.5pt,gray!50] (25,-10.5) -- (17,-10.5);
    \draw[-, line width=0.5pt,gray!50] (17,-10.5) -- (17,-6.5);
    \node at (26.5,-8.5) {$E_{1,k}$};


    \draw[-, line width=0.5pt] (u1k2) edge[bend left=20] (v110);
    \draw[-, line width=0.5pt] (v31) edge[bend right=15] (uk12);

    \draw [dash pattern=on 2pt off 2pt] (v19)-- (3,-7.5);
    \draw [dash pattern=on 2pt off 2pt] (v19)-- (3.5,-7.5);
    \draw [dash pattern=on 2pt off 2pt] (v19)-- (4,-7.5);

    \draw [dash pattern=on 2pt off 2pt] (v110)-- (5.5,-7.5);
    \draw [dash pattern=on 2pt off 2pt] (v110)-- (6,-7.5);
    \draw [dash pattern=on 2pt off 2pt] (v110)-- (6.5,-7.5);

    \draw [dash pattern=on 2pt off 2pt] (v111)-- (8.5,-7.5);
    \draw [dash pattern=on 2pt off 2pt] (v111)-- (9,-7.5);
    \draw [dash pattern=on 2pt off 2pt] (v111)-- (9.5,-7.5);

    \draw [dash pattern=on 2pt off 2pt] (v112)-- (11,-7.5);
    \draw [dash pattern=on 2pt off 2pt] (v112)-- (11.5,-7.5);
    \draw [dash pattern=on 2pt off 2pt] (v112)-- (12,-7.5);

    \draw [dash pattern=on 2pt off 2pt] (a13)-- (0,-7.5);

    \draw [dash pattern=on 2pt off 2pt] (v31)-- (3,-9.5);
    \draw [dash pattern=on 2pt off 2pt] (v31)-- (3.5,-9.5);
    \draw [dash pattern=on 2pt off 2pt] (v31)-- (4,-9.5);

    \draw [dash pattern=on 2pt off 2pt] (v32)-- (5.5,-9.5);
    \draw [dash pattern=on 2pt off 2pt] (v32)-- (6,-9.5);
    \draw [dash pattern=on 2pt off 2pt] (v32)-- (6.5,-9.5);

    \draw [dash pattern=on 2pt off 2pt] (v33)-- (8.5,-9.5);
    \draw [dash pattern=on 2pt off 2pt] (v33)-- (9,-9.5);
    \draw [dash pattern=on 2pt off 2pt] (v33)-- (9.5,-9.5);

    \draw [dash pattern=on 2pt off 2pt] (v34)-- (11,-9.5);
    \draw [dash pattern=on 2pt off 2pt] (v34)-- (11.5,-9.5);
    \draw [dash pattern=on 2pt off 2pt] (v34)-- (12,-9.5);

    \draw [dash pattern=on 2pt off 2pt] (a31)-- (0,-9.5);

    \draw [dash pattern=on 2pt off 2pt] (u1k1)-- (18,-6);
    \draw [dash pattern=on 2pt off 2pt] (u1k1)-- (18.5,-6);
    \draw [dash pattern=on 2pt off 2pt] (u1k1)-- (19,-6);

    \draw [dash pattern=on 2pt off 2pt] (u1k2)-- (19.5,-6);
    \draw [dash pattern=on 2pt off 2pt] (u1k2)-- (20,-6);
    \draw [dash pattern=on 2pt off 2pt] (u1k2)-- (20.5,-6);

    \draw [dash pattern=on 2pt off 2pt] (u1k3)-- (23,-6);
    \draw [dash pattern=on 2pt off 2pt] (u1k3)-- (23.5,-6);
    \draw [dash pattern=on 2pt off 2pt] (u1k3)-- (24,-6);

    \draw [dash pattern=on 2pt off 2pt] (uk11)-- (18,-11);
    \draw [dash pattern=on 2pt off 2pt] (uk11)-- (18.5,-11);
    \draw [dash pattern=on 2pt off 2pt] (uk11)-- (19,-11);

    \draw [dash pattern=on 2pt off 2pt] (uk12)-- (19.5,-11);
    \draw [dash pattern=on 2pt off 2pt] (uk12)-- (20,-11);
    \draw [dash pattern=on 2pt off 2pt] (uk12)-- (20.5,-11);

    \draw [dash pattern=on 2pt off 2pt] (uk13)-- (23,-11);
    \draw [dash pattern=on 2pt off 2pt] (uk13)-- (23.5,-11);
    \draw [dash pattern=on 2pt off 2pt] (uk13)-- (24,-11);

\end{tikzpicture}
\caption{An example of the connection between the $V_i$ and the $E_{l,m}$ gadgets used in the proof of Theorem~\ref{thm:hard-k-l-d}. The edge $u_2^{1,k}u_1^{k,1}$ of $E_{1,k}$ represents the edge $v_2^1v_1^k$ of $H$, where $v_2^1\in S_1$ and $v_1^k\in S_k$. The dotted edges are used to represent the edges connecting the vertices between the corresponding gadgets.}\label{fig:connection-gadgets-hard-k-l-d}
\end{figure}

To finalize the construction of the instance $I$, we need to specify the set of agents, as well as the functions $s_0$ and $t$ and the value of $\ell$. For the set of agents, let $A=\left\{\alpha_j^i:i\in[k] \text{ and } j\in [k]\setminus \{i \} \right\}$. Then, for any $i\in[k]$ and $j\in[k]\setminus \{i\}$, let $s_0(\alpha_j^i)=a_j^i$ and $t(\alpha_j^i)=t^i_j$. Finally, we set $\ell=3$. Intuitively, these are chosen so that the agents are obliged to move towards their goals in every turn, thus excluding the possibility for them to wait at any vertex at any given turn. Since we manage to show that the problem is hard even in this situation where the behavior of the agents is ``manageable'', the same holds true in the more general case. This finishes the construction of the instance $I$.

Before we move on with the reduction, we present some important observations. Observe first that the starting position of each agent is at a distance exactly $3$ from their target position.
Since in $I$ we have $\ell=3$, it follows that any feasible solution of $I$ will have makespan exactly $3$ and will be such that $s_1(\alpha_j^i)\in P_p^i$, for some $p\in [n]$, $s_2(\alpha_j^i)\in E_{l,m}$, for some $l,m\in [k]$ with $l<m$, and $s_3(\alpha_j^i) = t^i_j$, for every $i\in[k]$ and $j\in[k]\setminus \{i\}$.
Also, observe that in any feasible solution of $I$, we have that for every $i\in[k]$ there exists a unique $p\in [n]$ such that $s_1(\alpha_j^i)\in P_p^i$ for every $j\in[k]\setminus \{i\}$.
In fact, assume that $s_1$ is such that there exist $j< j'\in[k]\setminus \{i\}$ and $p<p'\in [n]$ with $v_1=s_1(\alpha_j^i)\in P_p^i$ and $v_2=s_1(\alpha_{j'}^i)\in P_{p'}^i$.
Then $\dist_G(v_1,v_2)\geq 2$. Since $d=1$ and $I$ is assumed to be feasible, we have that the positions of the agents during $s_1$ induce a path of $G$ that connects $v_1$ and $v_2$.
By $s_0$ and the construction of $G$, this path must necessarily include a vertex $a_j^i$ for some $j\in[k]\setminus \{i\}$; that is, there exists an $\alpha \in A$ such that $s_1(\alpha)=a_j^i$.
This is a contradiction to the makespan of the solution.

\paragraph*{The reduction.} We start by assuming that $I$ is a yes-instance of \MAPFCCShort and let $s_1,s_2,s_3$ be a feasible solution of $I$. It follows from the previous observations that for every $i\in[k]$ there exists a unique $p\in [n]$ such that $s_1(\alpha_j^i)\in P_p^i$ for every $j\in[k]\setminus \{i\}$; let us denote this $p$ by $p(i)$. Consider now the vertices $v_{p(i)}^i$, for every $i\in[k]$, of $H$. We claim that the subgraph of $H$ that is induced by these vertices is a clique on $k$ vertices. Indeed, assume that there are two indices $i<i'\in [k]$ such that $v_{p(i)}^iv_{p(i')}^{i'}\notin E(H)$. Consider now the vertices $v_{p(i),p(i')}^i$ and $v_{p(i'),p(i)}^{i'}$ of $V_i$ and $V_{i'}$ respectively. By the definition of $p(i)$ and $p(i')$, we have that both of these vertices are occupied by an agent, say $\alpha$ and $\beta$ respectively, during the turn $s_1$. Since $s_1,s_2,s_3$ is a feasible solution of $I$, and by the construction of $G$, we have that $s_2(\alpha)$ and $s_2(\beta)$ belong to $E_{i,i'}$. In particular, we have that $s_2(\alpha)=u^{i,p(i')}_{p(i)}$ and $s_2(\beta)=u^{i',p(i)}_{p(i')}$. This is a contradiction since, by its construction, the $E_{i,i'}$ gadget does not include these vertices since $v_{p(i)}^iv_{p(i')}^{i'}\notin E(H)$.

For the reverse direction, we assume that $H$ is a yes-instance of $k$-MCC. That is, for each $i\in[k]$, there exists a $p(i)\in [n]$ such that the vertices $\{v^i_{p(i)}:i\in[k]\}$ form a clique of $H$. For each $i\in[k]$ and $j\in[k]\setminus \{i\}$:
\begin{enumerate}
 \item We set $s_1(\alpha_j^i)=v^i_{p(i),j}$. Clearly, $s_0(\alpha_j^i)$ is a neighbor of $s_{1}(\alpha_j^i)$. Moreover, since $v^i_{p(i),k}$ is connected to $v^{i+1}_{p(i+1),1}$ for every $i\in [k-1]$, we have that the communication constraint is respected within $s_1$.
 \item Next, we set $s_2(\alpha_j^i)=u^{i,j}_{p(i)}$. Observe that for each $i\in[k]$ and $j\in[k]\setminus \{i\}$, we have that $v^i_{p(i),j}$ is a neighbor of $u^{i,j}_{p(i)}$ by the construction of $G$ and because the vertices $\{v^i_{p(i)}:i\in[k]\}$ form a clique of $H$. By the same arguments, and by the connectivity between the $E_{l,m}$ gadgets, we have that the connectivity constraint is also respected within $s_2$.
 \item Finally, we set $s_3(\alpha_j^i)=t_j^i$. It is trivial to check that $s_3(\alpha_j^i)$ is a neighbor of $s_2(\alpha_j^i)$ and that the connectivity constraint is respected within $s_3$.
\end{enumerate}
It follows that $s_1,s_2,s_3$ is a feasible solution of $I$. This completes the proof.
\end{proof}

At this point, we would like to observe the following detail in the construction of the previous proof. Consider the instance $I=\langle G,A,s_0,t,1,3\rangle$ that was constructed. If $I$ is a yes-instance of \MAPFCCShort and $s_1,s_2,s_3$ is a feasible solution, the agents find themselves on the vertices of the clique $Q$ according to $s_3$. The crucial observation is that since $Q$ is a clique, and all the vertices of $Q$ are accessible from all the vertices of the $E_{l,m}$ gadgets, we are free to assign whichever vertex of $Q$ to be the terminal of whichever agent, without affecting the feasibility of the solution. In other words, $s_1,s_2,s_3$ is also a feasible solution to the problem of \textsc{Anonymous Multiagent Path Finding with Communication Constraints} (\textsc{AMAPFCC}). This problem is defined exactly like \MAPFCCShort, the difference being that in \textsc{AMAPFCC} there is a set of terminals $T\subseteq V$, and the solution is feasible if all the vertices of $T$ are occupied by an agent (regardless of which agent it is) by the last turn of the schedule. This leads us to the following. 

\begin{corollary}\label{cor:AMAPFCC-is-hard}
    The \textsc{Anonymous Multiagent Path Finding with Communication Constraints} problem is \W[1]-hard parameterized by the number of agents, even for $\ell=3$ and $d=1$. 
\end{corollary}

This corollary appears in stark opposition to the usual \textsc{Anonymous Multiagent Path Finding} problem (with no communication constraint), in which finding an optimal solution in regards to the makespan is achievable in polynomial time~\cite{YL12,AY24}.

\section{Efficient Algorithms}
In this section, we present our FPT{} algorithms that solve the \MAPFCCShort problem. 

\subsection{Few Agents and Short Communication}
\begin{theorem}\label{thm:fpt-d-k-Delta}
  The \MAPFCC{} problem is in \FPT{} parameterized by the number of agents~$k$ plus the maximum degree $\Delta$ and the communication range $d$.
\end{theorem}
\begin{proof}
Let $\langle G, A, s_0, t,d, \ell\rangle$ be an instance of \MAPFCC.
The algorithm is as follows.
We build an auxiliary directed graph $H$ which has a vertex $u$ for every possible arrangement of the $k$ agents of $A$ into feasible ($d$-connected) positions.
Observe that $V(H)$ contains the vertices $u_s$ and $u_t$ which correspond to the initial and the final configurations of the agents on the vertices of $G$, respectively.
Two vertices $u_1,u_2\in V(H)$ are joined by the arc $(u_1,u_2)$ if and only if it is possible to move from the configuration represented by the vertex $u_1$ into the one represented by $u_2$ in one turn.
Clearly, $\langle G, A, s_0, t,d, \ell\rangle$ is a yes-instance of \MAPFCC if and only if there is a directed path in $H$ from $u_s$ to $u_t$, of length at most $\ell$.
That is, if $\dist_H(u_s,u_t) \le \ell$; which one can check, e.g., using BFS.
We will show that 
$|V(H)|\le \Delta^{O(dk)} kd^k n$, which suffices to prove the statement.

Let us consider an agent $a\in A$.
We will first count the different feasible arrangements of the $k$ agents of $A$ (i.e., the possible positions of the agents on the graph) such that $a$ is located at a vertex $u\in V(G)$. 
Since two agents are considered connected if they are in distance at most $d$ and we have $k$ agents, there must be a set $U$ of $kd$ vertices such that the induced subgraph $G[U]$ is connected and all agents are located in $U$.
It is known (see \cite[Proposition~5.1]{GuillemotM14}) that given $u$, there exist at most $\Delta^{O(dk)}$ sets $U$ of size $dk$ where $u \in U$ and $G[U]$ is connected.
We can also enumerate all such sets $U$ in $\Delta^{O(dk)}$ time.
Finally, we need to know the exact positions of the agents in $U$. 
Since the possible arrangement of the agents in $U$ are ${{kd} \choose k} k! \le kd^k$, we have that $|V(H)|\le \Delta^{O(dk)} kd^k n$. 
\end{proof}

In the next theorem, we use the algorithm presented in Theorem~\ref{thm:fpt-d-k-Delta} in order to efficiently solve the \MAPFCCShort problem on trees when we have a bounded number of agents and communication range. On a high level, the idea is that if the tree has a vertex $u$ of very high degree, then it will also contain a lot of ``useless'' vertices neighboring $u$. Indeed, at any given time, the pertinent neighbors of $u$ are at most equal to the number of agents, which is bounded. Once we remove these ``useless'' vertices, we are left with a graph of bounded maximum degree, at which point we may employ the previous algorithm. In the upcoming proof we formalize this idea.

\begin{theorem}\label{thm:fpt-d-k-tree}
  The \MAPFCC{} problem is in \FPT{} parameterized by the number of agents~$k$ plus the communication range $d$ when the input graph is a tree.
\end{theorem}
\begin{proof}
Let $I=\langle T, A, s_0, t,d, \ell\rangle$ be an instance of \MAPFCC, where $T$ is a tree.
The idea here is to create a new instance $I'=\langle T', A, s_0, t,d, \ell\rangle$ such that $I$ is a yes-instance if and only if the same is true for $I'$, where $T'$ is a tree of maximum degree $3k$.
In essence, we will prune the tree $T$ so that its maximum degree becomes bounded by $3k$.
Furthermore, we show that, given a valid schedule for~$I$, we can adapt the movements of the agents appropriately to obtain a valid schedule for~$I'$ of the same length.
Once this is done, it suffices to apply the algorithm provided in Theorem~\ref{thm:fpt-d-k-Delta} to decide whether $I'$ is a yes-instance or not.

\begin{figure}[!t]
\scalebox{0.8}{
\subfloat[The tree $T$ (unpruned). The paths connecting the vertices in the left (right resp.) blob to $u$ correspond to $P^-(u)$ ($P^+(u)$ resp.). The lavender blob corresponds to the tree $T_u$.]{
\begin{tikzpicture}[inner sep=0.7mm]
  \filldraw[fill=red!30, draw=black, thick, fill opacity=0.2] plot[smooth cycle, tension=0.7] coordinates {
        (-5,-2.5) (-4,1) (-2,3) (0,2.8) (-0.8,1.2) (-1.8,0.2) (-2,-2.3)};
  \node[] at (-1.9,2.5) {$k$ paths};

  \filldraw[fill=red!30, draw=black, thick, fill opacity=0.2] plot[smooth cycle, tension=0.7] coordinates {
        (3.5,-2) (3.2,0) (2.7,3) (1.2,2.8) (1.5,0.1) (2,-1.8)};
  \node[] at (2,2.5) {$k$ paths};

  \filldraw[fill=blue!30, draw=black, thick, fill opacity=0.2] plot[smooth cycle, tension=0.7] coordinates {
        (-1.2,-0.7) (0,0.5) (1.2,-0.6) (1.4,-2) (2,-3) (0.2,-2.8) (-2,-3.5) (-1.5,-2)};
  \node[] at (0,-2) {$>k$ paths};
  \node[] at (0.5,-3) {$T_u$};

    \draw[] plot[smooth] coordinates {(-4,-2) (-3.5, -1) (-1.7, -1) (-1,-0.1) (0,0)};
    \draw[] plot[smooth] coordinates {(-3,1) (-2, 1) (-1.7, -0.1) (-1,0.3) (0,0)};
    \draw[] plot[smooth] coordinates {(-1,2) (-0.2, 1.7) (-1.1, 1) (-0.1,0.3) (0,0)};

    \draw[] plot[smooth] coordinates {(0,0) (0.4, 1) (1.3, 0.8) (1.7,2) (2,2)};
    \draw[] plot[smooth] coordinates {(0,0) (1, -0.5) (1.7, 0) (2.3,0.8) (2.9,0)};
    \draw[] plot[smooth] coordinates {(1, -0.5) (2,-0.7) (3.2,-1.5)};

 \node[draw, circle, line width=1pt, fill=white](u) at (0,0)[label=right:$u$] {};
 \node[draw, circle, line width=1pt, fill=white](s1) at (-4,-2)[] {};
 \node[] at (-1,-2.4) {\includegraphics[scale=0.05]{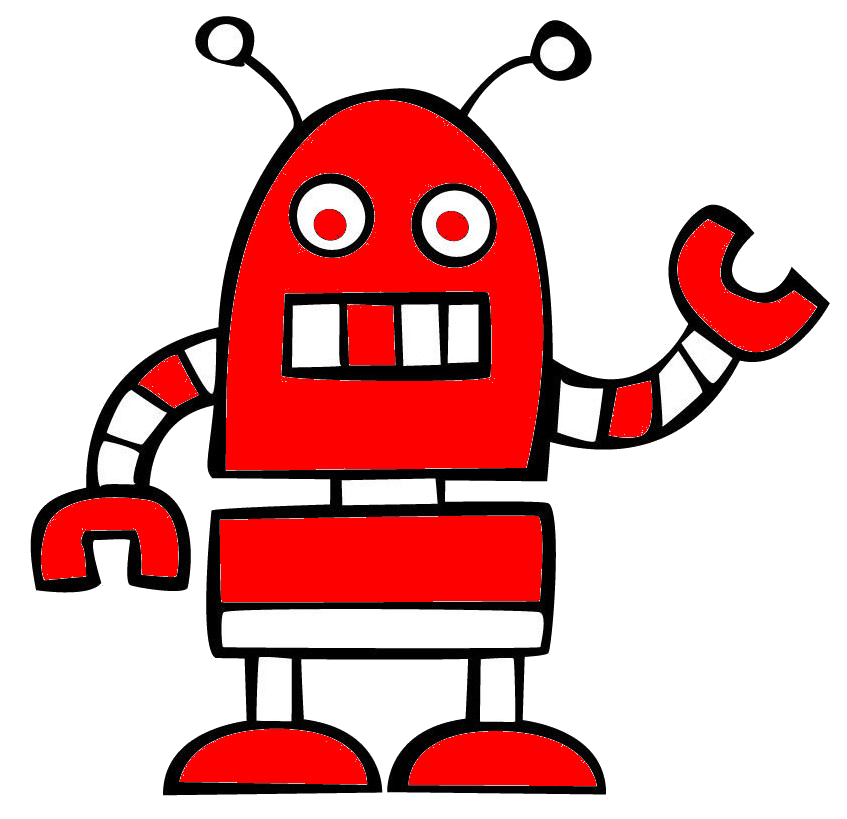}};
 \node[draw, circle, line width=1pt, fill=white](s2) at (-3,1)[] {};
 \node[left = 0.01cm of s2] (r2) {\includegraphics[scale=0.05]{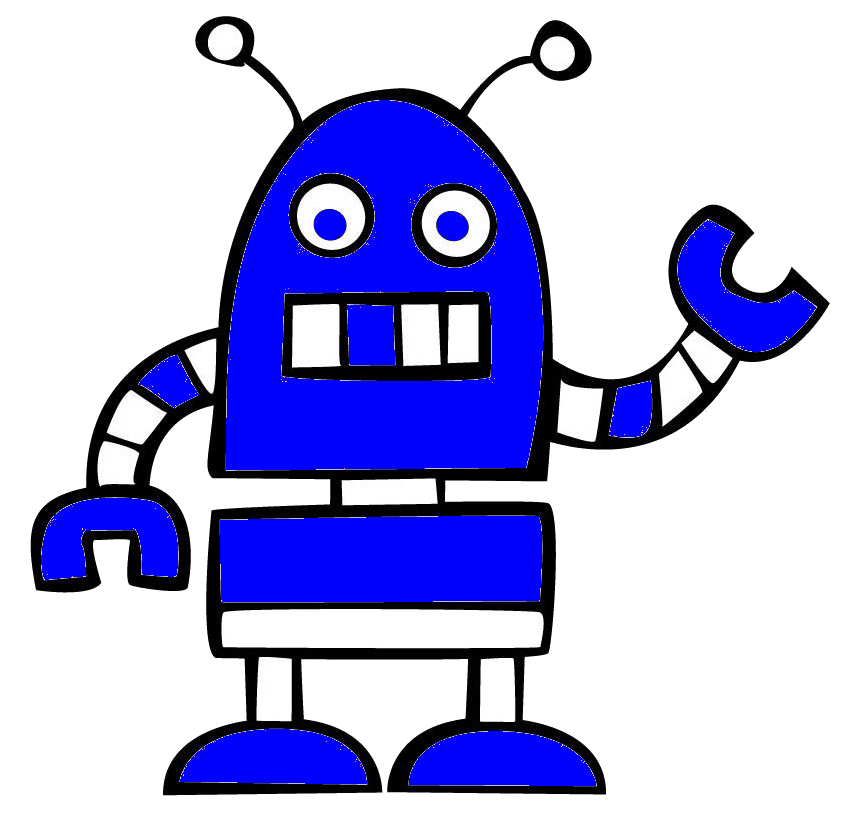}};
 \node[draw, circle, line width=1pt, fill=white](sk) at (-1,2)[] {};
 \node[left = 0.01cm of sk] (r3) {\includegraphics[scale=0.05]{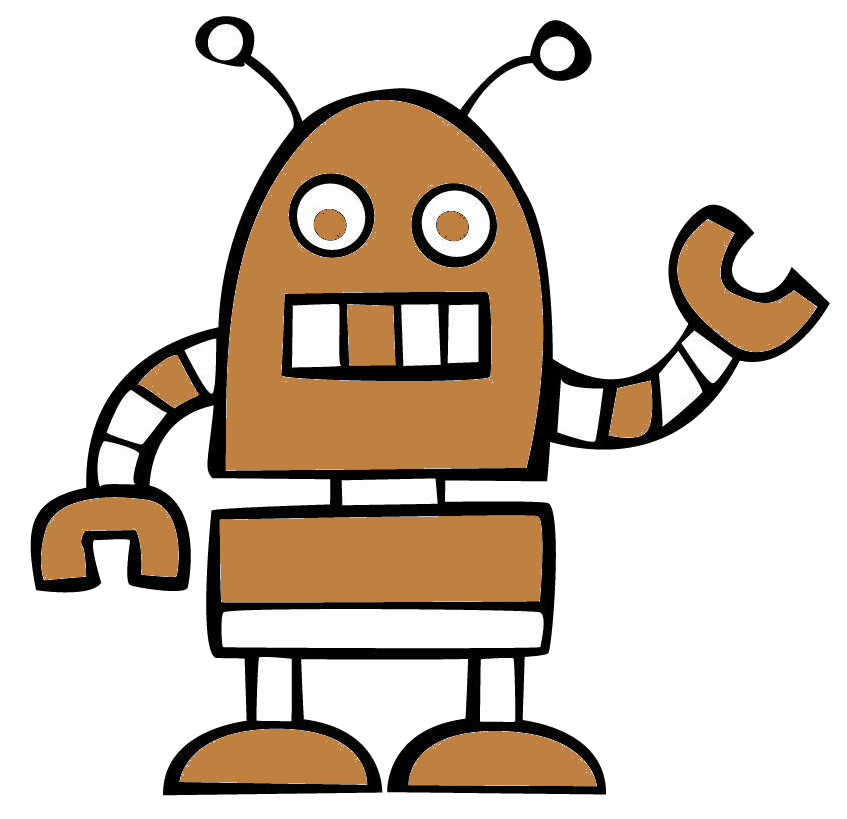}};

 \node[draw, circle, line width=1pt, fill=red](t1) at (2,2)[] {};
 \node[draw, circle, line width=1pt, fill=blue](t2) at (2.9,0)[] {};
 \node[draw, circle, line width=1pt, fill=brown](tk) at (3.2,-1.5)[] {};

 \node[draw, circle, line width=1pt, fill=white](v1) at (-1,-1)[] {};
 \node[draw, circle, line width=1pt, fill=white](v2) at (-0.5,-1)[] {};
 \node[draw, circle, line width=1pt, fill=white](v3) at (0,-1)[] {};
 \node[draw, circle, line width=1pt, fill=white](v4) at (0.5,-1)[] {};
 \node[draw, circle, line width=1pt, fill=white](v5) at (1,-1)[] {};

 \draw[-, line width=0.5pt]  (u) -- (v1);
 \draw[-, line width=0.5pt]  (u) -- (v2);
 \draw[-, line width=0.5pt]  (u) -- (v3);
 \draw[-, line width=0.5pt]  (u) -- (v4);
 \draw[-, line width=0.5pt]  (u) -- (v5);

\end{tikzpicture}
}\hspace{20pt}
\subfloat[The tree $\hat{T}$ (pruned). The red agent is placed at one of the $k$ remaining neighbors of $u$ for the same number of turns as it would spend in the pruned part.]{
\begin{tikzpicture}[inner sep=0.7mm]
  \filldraw[fill=red!30, draw=black, thick, fill opacity=0.2] plot[smooth cycle, tension=0.7] coordinates {
        (-5,-2.5) (-4,1) (-2,3) (0,2.8) (-0.8,1.2) (-1.8,0.2) (-2,-2.3)};
  \node[] at (-1.9,2.5) {$k$ paths};

  \filldraw[fill=red!30, draw=black, thick, fill opacity=0.2] plot[smooth cycle, tension=0.7] coordinates {
        (3.5,-2) (3.2,0) (2.7,3) (1.2,2.8) (1.5,0.1) (2,-1.8)};
  \node[] at (2,2.5) {$k$ paths};

  \filldraw[fill=red!30, draw=black, thick, fill opacity=0.2] plot[smooth cycle, tension=0.7] coordinates {
        (-1.2,-0.7) (0,0.5) (0.5, 0) (0.2,-0.6) (0.3,-1.1) (-1.4,-1.2)};
  \node[] at (-0.5,-1.5) {$k$ vertices};

    \draw[] plot[smooth] coordinates {(-4,-2) (-3.5, -1) (-1.7, -1) (-1,-0.1) (0,0)};
    \draw[] plot[smooth] coordinates {(-3,1) (-2, 1) (-1.7, -0.1) (-1,0.3) (0,0)};
    \draw[] plot[smooth] coordinates {(-1,2) (-0.2, 1.7) (-1.1, 1) (-0.1,0.3) (0,0)};

    \draw[] plot[smooth] coordinates {(0,0) (0.4, 1) (1.3, 0.8) (1.7,2) (2,2)};
    \draw[] plot[smooth] coordinates {(0,0) (1, -0.5) (1.7, 0) (2.3,0.8) (2.9,0)};
    \draw[] plot[smooth] coordinates {(1, -0.5) (2,-0.7) (3.2,-1.5)};

 \node[draw, circle, line width=1pt, fill=white](u) at (0,0)[label=right:$u$] {};
 \node[draw, circle, line width=1pt, fill=white](s1) at (-4,-2)[] {};

 \node[] at (-1.3,-1) {\includegraphics[scale=0.05]{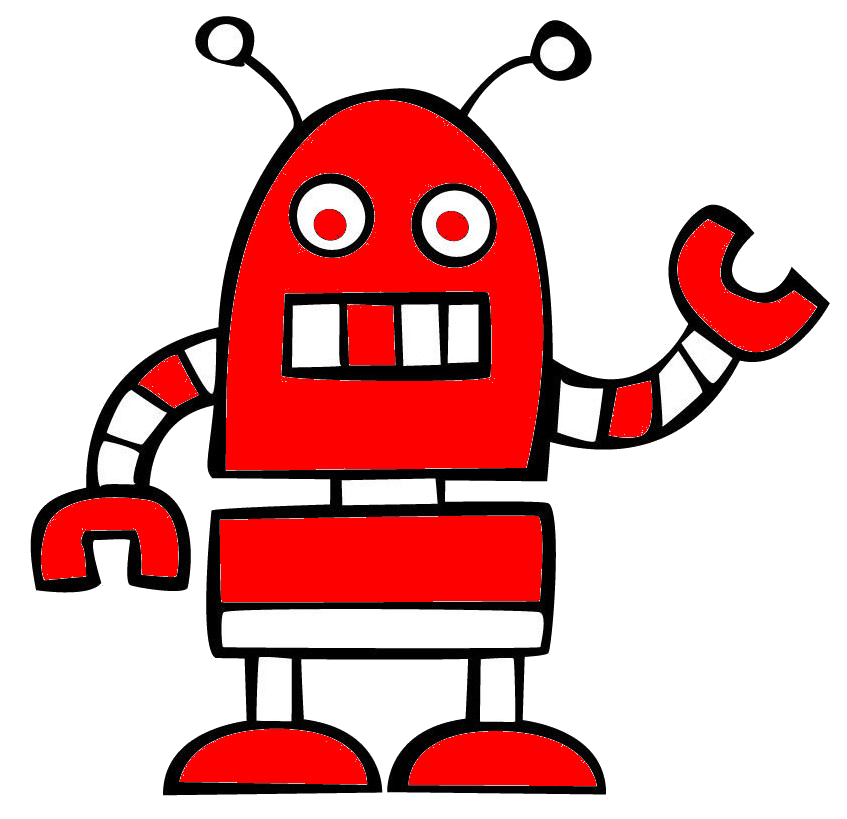}};
 
 \node[draw, circle, line width=1pt, fill=white](s2) at (-3,1)[] {};
 \node[left = 0.01cm of s2] (r2) {\includegraphics[scale=0.05]{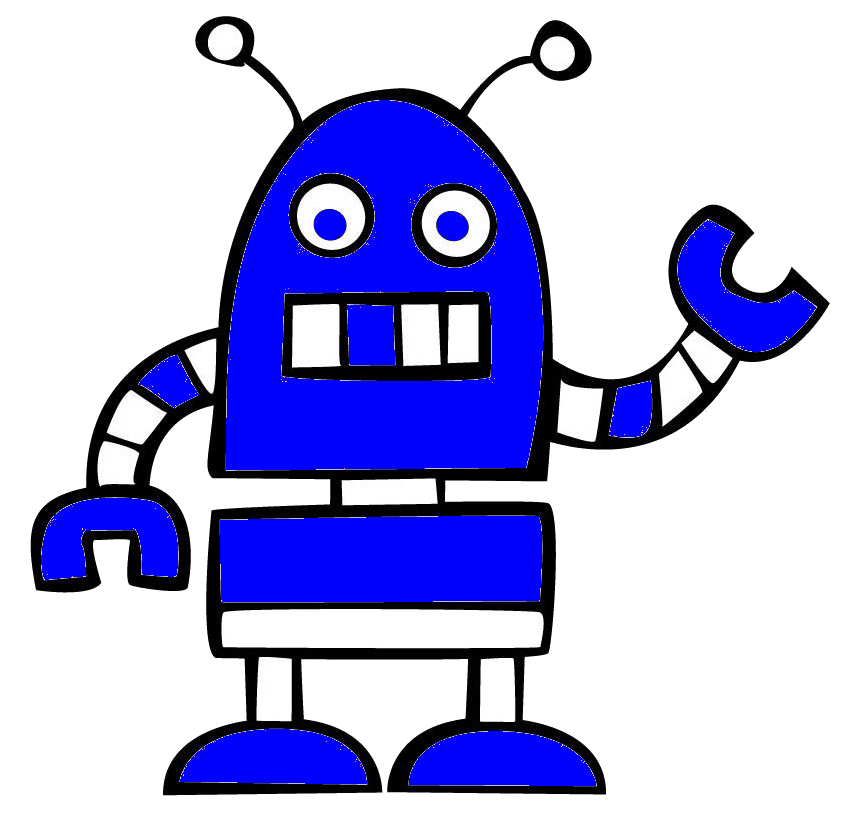}};
 \node[draw, circle, line width=1pt, fill=white](sk) at (-1,2)[] {};
 \node[left = 0.01cm of sk] (r3) {\includegraphics[scale=0.05]{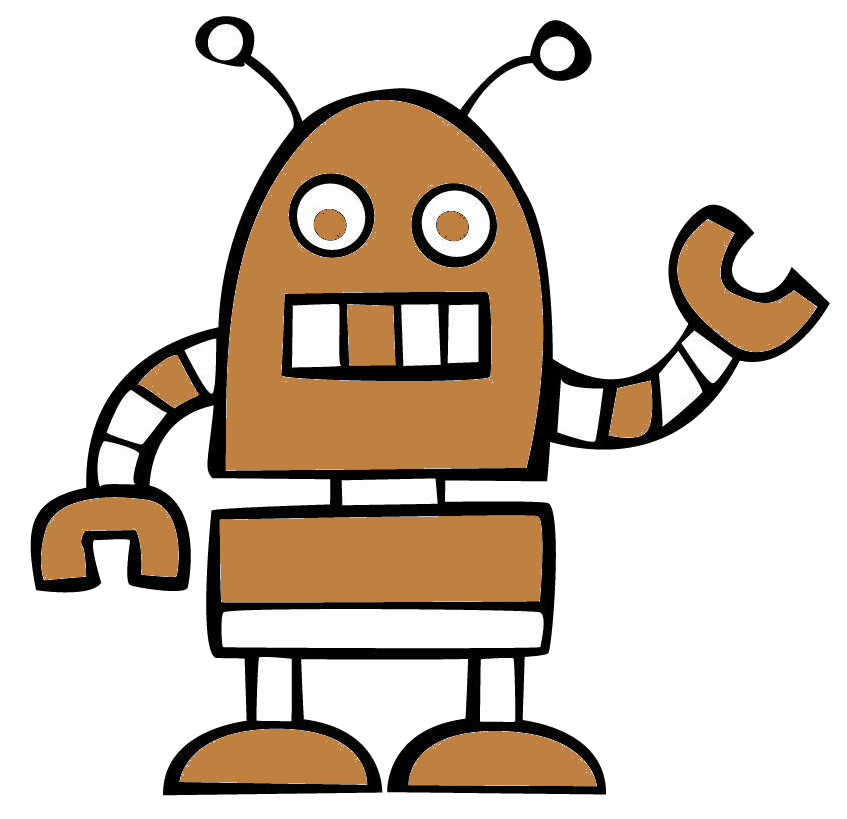}};

 \node[draw, circle, line width=1pt, fill=red](t1) at (2,2)[] {};
 \node[draw, circle, line width=1pt, fill=blue](t2) at (2.9,0)[] {};
 \node[draw, circle, line width=1pt, fill=brown](tk) at (3.2,-1.5)[] {};

 \node[draw, circle, line width=1pt, fill=white](v1) at (-1,-1)[] {};
 \node[draw, circle, line width=1pt, fill=white](v2) at (-0.5,-1)[] {};
 \node[draw, circle, line width=1pt, fill=white](v3) at (0,-1)[] {};

 \draw[-, line width=0.5pt]  (u) -- (v1);
 \draw[-, line width=0.5pt]  (u) -- (v2);
 \draw[-, line width=0.5pt]  (u) -- (v3);
\end{tikzpicture}
}
}
\caption{An abstract illustration of the trees $T$, $T_u$ and $\hat{T}$ employed in the proof of Theorem~\ref{thm:fpt-d-k-tree}. The colors on the agents and the vertices are used to encode the terminal vertex of each agent. In each subfigure, the left (right resp.) blob contains the starting (terminal resp.) positions of the agents. Both subfigures depict the instance during the same turn $i$, when the red agent has moved to some vertex inside $T_u$ in $T$ and some neighbor of $u$ in $\hat{T}$.}\label{fig-pruning}
\end{figure}

We now describe the pruning procedure (see Figure~\ref{fig-pruning}). If $\Delta(T)\leq 3k$, then we are done.
So, let $u\in V(T)$ be such that $d_T(u)>3k$.
For each $j\in [k]$, we define $P^-_{j,u}$ as the simple path of $T$ that connects $s_0(a_j)$ to $u$.
Similarly, let $P^+_{j,u}$ be the simple path of $T$ that connects $u$ to $t(a_j)$.
Observe that in the case where $u=s_0(a_j)$ ($u=t(a_j)$ respectively), for some $j\in [k]$, then $P^-_{j,u}=\emptyset$ ($P^+_{j,u}=\emptyset$ respectively).
Then, we define $P^-(u)=\{P^-_{1,u},\dots,P^-_{k,u}\}$ and $P^+(u)=\{P^+_{1,u},\dots,P^+_{k,u}\}$.
Intuitively, the set $P^-(u)$ contains all the paths of $T$ that will be relevant for the agents to reach $u$ from their initial positions.
The set $P^+(u)$ contains all the paths of $T$ that will be relevant for the agents to reach their targets from $u$.
Finally, let $V_u=N(u)\cap (P^-(u)\cup P^+(u))$.
That is, $V_u$ contains the neighbors of $u$ that are relevant with respect to the paths mentioned above.
Observe that $|V_u|\leq 2k$.
Let $T_u$ be the subtree of $T[V\setminus V_u]$ that contains $u$.
Since $d_T(u)>3k$, we have that $T_u$ contains at least the vertices $N(u)\setminus V_u$.
We are now ready to describe the pruning that we perform: starting from $T$, remove all the vertices of $T_u$, except from $u$ and $k$ of its neighbors $v_1,\dots,v_k$ in $T_u$; let $\hat{T}$ be the resulting graph.

\begin{claim}
If $I=\langle T, A, s_0, t,d, \ell\rangle$ is a yes-instance, then $\hat{I}=\langle \hat{T}, A, s_0, t,d, \ell\rangle$ is also a yes-instance.
\end{claim}
\begin{proofclaim}
Assume that $s=(s_1,\dots,s_\ell)$ is a feasible solution of $I$; we construct a feasible solution $s'=s'_1,\dots,s'_\ell$ of $I'$.
We start by setting $s'_0(a)=s_0(a)$ for every $a\in A$.
Then, for every $i\in[\ell]$ and for every $j\in [k]$, we set
\[
s'_i(a_j) =
     \begin{cases}
       v_j, &\quad\text{if } s_j(a_j)\in T_u\setminus u\\
       s_i(a_j), &\quad\text{otherwise.}\\
     \end{cases}
\]

First, we need to show that $s'_i(a)$ is a neighbor of $s'_{i-1}(a)$ for every $a\in A$. Let $a\in A$ and $i\in [\ell]$.
We distinguish the following cases:
\begin{description}
 \item[$\boldsymbol{s_i(a)\in T_u\setminus u}$ \textbf{and} $\boldsymbol{s_{i-1}(a)\in T_u\setminus u}$] Then $s'_i(a)=s'_{i-1}(a)=v_j$ for some $j\in [k]$.
 \item[$\boldsymbol{s_i(a)\in T_u\setminus u}$ \textbf{and} $\boldsymbol{s_{i-1}(a)\notin T_u\setminus u}$] Since $s$ is a feasible solution, we have that $s_{i-1}(a)=u$. Thus, $s'_i(a)=v_j$ for some $j\in [k]$ and $s'_{i-1}(a)=s_{i-1}(a)=u$.
 \item[$\boldsymbol{s_i(a)\notin T_u\setminus u}$ \textbf{and} $\boldsymbol{s_{i-1}(a)\in T_u\setminus u}$] Is analogous to the previous one.
 \item[$\boldsymbol{s_i(a)\notin T_u\setminus u}$ \textbf{and} $\boldsymbol{s_{i-1}(a)\notin T_u\setminus u}$] Then $s'_i(a)=s_i(a)$ and $s'_{i-1}(a)=s_{i-1}(a)$ and the feasibility of $s'$ is guaranteed by the feasibility of $s$.
\end{description}

Next, we need to show that $s'_\ell(a_j)=t(a_j)$ for every $j\in [k]$. This follows directly from the definitions of $s'$ and the paths $P^-_{j,u}$ and $P^+_{j,u}$ and the fact that~$s$ is a feasible schedule.

Finally, we will ensure that the connectivity constraint is preserved in $s'$.
Let $D$ and $D'$ be the communication graphs of $T$ and $\hat{T}$ respectively.
We will show that $D'[ \{ s'_i(a)\mid a\in A \}]$ is connected for every $i\in[\ell]$.
Assume that this is not true, and let $l \in [\ell]$ be one turn during which the connectivity constraint fails in $s'$.
That is, $D'[\{ s'_l(a)\mid a\in A \}]$ contains at least two connected components.
Since $s$ is a feasible solution of $I$, it follows that there exist two agents $b$ and $c$ such that $\dist_T(s_l(b),s_l(c))\leq d$ but $\dist_{\hat{T}}(s'_l(b),s'_l(c)) > d$.
By the definition of $s'$, we have that at least one of the agents $b$ and $c$ is located in $T_u\setminus u$ according to $s'_l$.
We distinguish the following cases for the values of $d\geq 2$ (we deal with the case where $d=1$ afterwards):

\begin{description}
 \item[$\boldsymbol{s'_l(b)\in T_u\setminus u}$ \textbf{and} $\boldsymbol{s'_l(c)\notin T_u\setminus u}$] It holds that $\dist_{T'}(s'_l(b),s'_l(c))\leq \dist_{T}(s_l(b),s_l(c))-1\leq d-1$.
 \item[$\boldsymbol{s'_l(b)\in T_u\setminus u}$ \textbf{and} $\boldsymbol{s'_l(c)\in T_u\setminus u}$] In this case, there exist two vertices $v_x$ and $v_z$, for some $x<z\in [k]$ which are leafs attached to $u$ in $T'$, such that $s'_l(b)=v_x$ and $s'_l(c)=v_z$. Clearly, $\dist(v_x,v_z)=2\leq d$.
\end{description}

Lastly, we deal with the particular case of $d=1$.
We additionally assume that~$\ell$ is optimal, that is, there is no shorter feasible schedule.
We claim that there exists an agent $a^*\in A$ such that $s'_l(a^*)=s_l(a^*)=u$. Let us assume that this is not true. We distinguish the following two sub-cases:
  \begin{itemize}
    \item The agent $b$ is such that $s_l(b)=s'_l(b)=x\in V(P^-(u))\cup V(P^+(u))$. In this case, and since at least one of the agents $b$ and $c$ must be in $T_u\setminus u$ during the turn $l$, the existence of $a^*$ is guaranteed by the feasibility of $I$ for $d=1$.
    \item Every agent $a\in A$ is such that $s_l(a)\in T_u$. Consider $l_0<l$, the last turn that the vertex $u$ was occupied before the turn $l$, say by the agent $e$. Also, let $l<l_1\leq \ell$ be the first turn that the vertex $u$ will be occupied after the turn $l$, say by the agent $g$. Let us also consider what happens in $T'$ according to $s'$ during these turns. We have that at the turn $l_0$ all the agents of $A$ are on the leafs attached to $u$ except for the agent $e$ which is on $u$. Then, on the turn $l_0+1$, all the agents of $A$ are on the leafs attached to $u$, and remain there until the turn $l_1$, where the agent $g$ moves to $u$. We then modify $s'$ by setting $s'_{l_0+1}(u)=g$. If $l_0+1=l$, we have a contradiction as $u$ is assumed to be empty during the turn $l$ according to $s'$. Thus, $l_0+1<l$. But in this case, the makespan of the modified $s'$ is smaller than $\ell$ (note that $s'$ is a solution for~$I$ as well). But $\ell$ was assumed to be optimal, leading to a contradiction.
 \end{itemize}
In all the cases above we are lead to a contradiction, proving that the connectivity constraint of $s$ is indeed preserved by~$s'$.
\end{proofclaim}

For the other direction of the equivalence, it holds trivially since $\hat{T}$ is a subgraph of $T$.
Therefore, we apply the above pruning procedure exhaustively, that is, as long as there is a vertex of degree at least $3k+1$.
In this way, we obtain the tree~$T'$ and the instance~$I'$.
\end{proof}

\subsection{Tree-like Structures}
In the previous section we presented an efficient algorithm to solve \MAPFCCShort{} when we have a bounded number of agents, communication distance and maximum degree. Furthermore, we can relax the requirement for a bounded maximum degree if we consider graphs that are trees. The next natural question to pose is whether this remains true when we consider graphs that are ``close to being trees'', i.e., have bounded treewidth. In the next theorem we answer positively to this question. This is done by capturing the dynamicity of the problem through a static auxiliary graph (the graph $G_I$, see Figure~\ref{fig:expanded} for an example). We then prove that if the original graph has bounded treewidth, the same holds true for the auxiliary graph. Finally, we encode the \MAPFCCShort problem on the auxiliary graph in an \MSOtwo formula, and employ Courcelle's theorem~\cite{Courcelle90} to construct an \FPT~algorithm.  

\begin{theorem}\label{thm:fpt-tw-d-l}
	The \MAPFCCShort problem is in \FPT{} parameterized by the treewidth~$\tw$ of $G$ plus the makespan~$\ell$ and the communication range~$d$.
\end{theorem}
\begin{proof}
Let $I = \langle G, A, s_0 , t, d, \ell\rangle$ be an instance of \MAPFCCShort.
Our goal is to construct an auxiliary graph~$G_I$ with special vertex and edge labels, such that (i) the treewidth of $G_I$ is at most~$3 \ell \cdot \tw$ whenever $I$ is a yes-instance, and (ii) the existence of a solution can be expressed by an \MSOtwo sentence over $G_I$.
The claim then follows by testing the treewidth of $G_I$ followed by a standard use of Courcelle's theorem~\cite{Courcelle90}.

Let $G = (V,E)$ be the graph of the input instance.
We start by constructing a labeled auxiliary graph $G_I$ (see Figure~\ref{fig:expanded}). 
Before we do this, allow us to provide a short intuition on the graph $G_I$. In essence, our goal with $G_I$ is to capture the dynamicity of the problem. This is done by considering consecutive ``snapshots'' of $G$. These snapshots are then linked by considering which vertices are available to an agent in the next turn, according to its current position. Finally, for each agent $a$, we add an edge between the starting position of $a$ in the first snapshot and the terminal vertex of $a$ in the last snapshot. This is done to encode the information of where does each agent want to go into the bags of the tree-decomposition of $G_I$. 

We now proceed with formally defining $G_I$.
Its vertex set is composed of sets $V_0, \dots, V_\ell$ where $V_i$ contains one copy of each vertex of $V$, that is, $V_i = \{v_i \mid v \in V\}$ (we denote by $v_i$ the copy of the vertex $v$ in $V_i$).
We refer to these sets as \emph{layers} and we give all the vertices in $V_i$ a vertex label $\mathsf{vertex_i}$.
The graph $G_I$ contains four different types of edges with four distinct edge labels defined as follows.
\begin{enumerate}
\item For every $v \in V$ and $i \in [\ell]$, we add to $G_I$ the edge $v_{i-1}v_i$ with an edge label $\mathsf{copy}$.
\item For every $uv \in E$ and $i \in \intzero{\ell}$, we add to $G_I$ the edge $u_i v_i$ with an edge label $\mathsf{communication}$.
\item For every $uv \in E$ and $i \in [\ell]$, we add to $G_I$ the edges $u_{i-1}v_i$ and $v_{i-1}u_i$ with an edge label $\mathsf{cross}$.
\item Finally, for every agent $a \in A$, we add to $G_I$ the edge $s_0(a)_0 t(a)_\ell$ with an edge label $\mathsf{agent}$.
\end{enumerate}

\begin{figure}[!t]
\subfloat[The graph $G$.]{
\scalebox{1}{
\begin{tikzpicture}[inner sep=0.7mm]

\node[draw,black,circle,line width=1pt,fill=white] (v1) at (1.5,2){};
\node[draw,black,circle,line width=1pt,fill=white] (v2) at (1.5,0.8){};
\node[draw,black,circle,line width=1pt,fill=brown] (v3) at (2.5,0){};
\node[draw,black,circle,line width=1pt,fill=white] (v4) at (0.5,0){};

\node[left = 0.01cm of v4] (r1) {\includegraphics[scale=0.05]{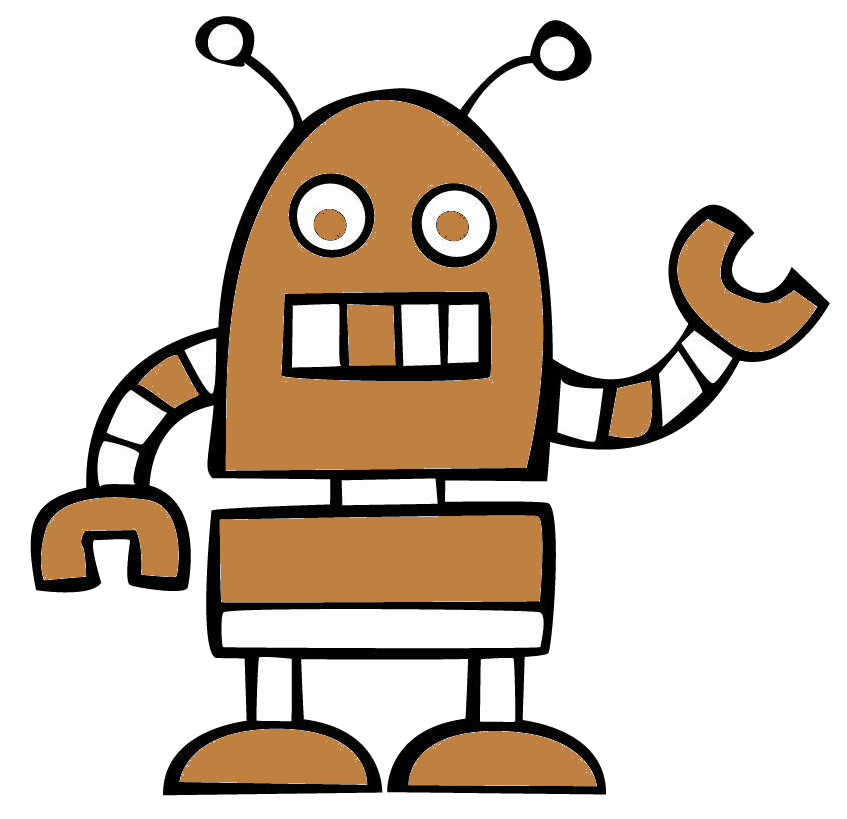}};

\draw[line width=1pt,black] (v1) -- (v2);
\draw[line width=1pt,black] (v2) -- (v3);
\draw[line width=1pt,black] (v2) -- (v4);

\end{tikzpicture}
}
}
\hspace{20pt}
\subfloat[The graph $G_I$.]{
\scalebox{0.9}{
\begin{tikzpicture}[inner sep=0.7mm]

\node[draw,black,circle,line width=1pt,fill=white] (v1) at (1.5,2){};
\node[draw,black,circle,line width=1pt,fill=white] (v2) at (1.5,0.8){};
\node[draw,black,circle,line width=1pt,fill=white] (v3) at (2.5,0){};
\node[draw,black,circle,line width=1pt,fill=white] (v4) at (0.5,0){};

\node[draw,black,circle,line width=1pt,fill=white] (v5) at (4.5,2){};
\node[draw,black,circle,line width=1pt,fill=white] (v6) at (4.5,0.8){};
\node[draw,black,circle,line width=1pt,fill=white] (v7) at (5.5,0){};
\node[draw,black,circle,line width=1pt,fill=white] (v8) at (3.5,0){};

\node[draw,black,circle,line width=1pt,fill=white] (v9) at (7.5,2){};
\node[draw,black,circle,line width=1pt,fill=white] (v10) at (7.5,0.8){};
\node[draw,black,circle,line width=1pt,fill=white] (v11) at (8.5,0){};
\node[draw,black,circle,line width=1pt,fill=white] (v12) at (6.5,0){};

\node[draw,black,circle,line width=1pt,fill=white] (v13) at (10.5,2){};
\node[draw,black,circle,line width=1pt,fill=white] (v14) at (10.5,0.8){};
\node[draw,black,circle,line width=1pt,fill=white] (v15) at (11.5,0){};
\node[draw,black,circle,line width=1pt,fill=white] (v16) at (9.5,0){};

\draw[-,line width=1pt,red] (v1) -- (v2);
\draw[-,line width=1pt,red] (v2) -- (v3);
\draw[-,line width=1pt,red] (v2) -- (v4);

\draw[-,line width=1pt,red] (v5) -- (v6);
\draw[-,line width=1pt,red] (v6) -- (v7);
\draw[-,line width=1pt,red] (v6) -- (v8);

\draw[-,line width=1pt,red] (v9) -- (v10);
\draw[-,line width=1pt,red] (v10) -- (v11);
\draw[-,line width=1pt,red] (v10) -- (v12);

\draw[-,line width=1pt,red] (v13) -- (v14);
\draw[-,line width=1pt,red] (v14) -- (v15);
\draw[-,line width=1pt,red] (v14) -- (v16);


\draw[-,line width=1pt,blue] (v1) -- (v5);

\draw[-,line width=1pt,blue] (v5) -- (v9);
\draw[-,line width=1pt,blue] (v9) -- (v13);

\draw[-,line width=1pt,blue] (v2) -- (v6);

\draw[-,line width=1pt,blue] (v6) -- (v10);
\draw[-,line width=1pt,blue] (v10) -- (v14);

\draw[-,line width=1,blue] (v4) to[out=-40,in=-140] (v8);
\draw[-,line width=1,blue] (v3) to[out=-40,in=-140] (v7);

\draw[-,line width=1,blue] (v8) to[out=-40,in=-140] (v12);
\draw[-,line width=1,blue] (v7) to[out=-40,in=-140] (v11);

\draw[-,line width=1,blue] (v11) to[out=-40,in=-140] (v15);
\draw[-,line width=1,blue] (v12) to[out=-40,in=-140] (v16);


\draw[-,line width=1pt,black] (v1) -- (v6);

\draw[-,line width=1pt,black] (v2) -- (v5);
\draw[-,line width=1pt,black] (v2) -- (v8);
\draw[-,line width=1pt,black] (v2) -- (v7);

\draw[-,line width=1pt,black] (v3) -- (v6);
\draw[-,line width=1pt,black] (v4) -- (v6);

\draw[-,line width=1pt,black] (v5) -- (v10);
\draw[-,line width=1pt,black] (v6) -- (v9);
\draw[-,line width=1pt,black] (v6) -- (v12);
\draw[-,line width=1pt,black] (v6) -- (v11);
\draw[-,line width=1pt,black] (v7) -- (v10);
\draw[-,line width=1pt,black] (v8) -- (v10);

\draw[-,line width=1pt,black] (v9) -- (v14);
\draw[-,line width=1pt,black] (v10) -- (v13);
\draw[-,line width=1pt,black] (v10) -- (v16);
\draw[-,line width=1pt,black] (v10) -- (v15);
\draw[-,line width=1pt,black] (v11) -- (v14);
\draw[-,line width=1pt,black] (v12) -- (v14);

\draw[-,line width=1,brown] (v4) to[out=-40,in=-140] (v15);

\draw [decorate,decoration={brace, amplitude=10pt},xshift=0pt,yshift=5pt] (0.3,2) -- (2.7,2) node [black,midway,yshift=-0.7cm] {};
\node[] at (1.5,3)  []{\large layer $0$};

\draw [decorate,decoration={brace, amplitude=10pt},xshift=0pt,yshift=5pt] (3.3,2) -- (5.7,2) node [black,midway,yshift=-0.7cm] {};
\node[] at (4.5,3)  []{\large layer $1$};

\draw [decorate,decoration={brace, amplitude=10pt},xshift=0pt,yshift=5pt] (6.3,2) -- (8.7,2) node [black,midway,yshift=-0.7cm] {};
\node[] at (7.5,3)  []{\large layer $2$};

\draw [decorate,decoration={brace, amplitude=10pt},xshift=0pt,yshift=5pt] (9.3,2) -- (11.7,2) node [black,midway,yshift=-0.7cm] {};
\node[] at (10.5,3)  []{\large layer $3$};
\end{tikzpicture}
}
}
\caption{An example of a graph $G$ and the corresponding graph $G_I$ for $\ell=3$, used in the proof of Theorem~\ref{thm:fpt-tw-d-l}. In $G$, the brown agent is set next to that agents' staring vertex and the brown vertex is its target. The colors on the edges of $G_I$ are used to represent the various labels we employ. In particular, blue edges are labeled $\mathsf{copy}$, red edges are labeled $\mathsf{communication}$, black edges are labeled $\mathsf{cross}$ and the brown edge is labeled $\mathsf{agent}$.}\label{fig:expanded}
\end{figure}
Observe that the first three types of edges correspond exactly to the strong product of $G$ with a path of length~$\ell$.
In a sense, the construction combines the time-expanded graphs used previously for \MAPFShort~\cite{FKKMO24} with the augmented graphs considered for edge-disjoint paths~\cite{ZTN00,GOR21}.
Importantly, we observe that the existence of a solution is equivalent to the existence of a set of vertex-disjoint paths in~$G_I$ with some additional properties.

\begin{observation}\label{obs:yes-instance-paths}
The instance $I$ is a yes-instance if and only if there exists a set of vertex-disjoint paths ${\mathcal{P} = \{P_a \mid a \in A \}}$ such that:
\begin{enumerate}
\item each $P_a$ contains exactly one vertex from each layer,\label{cond:yes-instance-paths-1}
\item the endpoints of each $P_a$ are the vertices $s_0(a)_0$ and $t(a)_\ell$,\label{cond:yes-instance-paths-2}
\item there are no two paths $P_a$ and $P_b$, vertices $u,v \in V$  and $i \in [\ell]$ such that $P_a$ contains the edge $u_{i-1}v_i$ and $P_b$ contains the edge $v_{i-1}u_i$, and\label{cond:yes-instance-paths-3}
\item for each $i \in \intzero{\ell}$, the set of vertices $W_i \subseteq V_i$ visited by paths from~$\mathcal{P}$ forms a $d$-connected set in~$G[V_i]$.\label{cond:yes-instance-paths-4}
\end{enumerate}
\end{observation}

We show later how to translate these properties into \MSOtwo predicates.
Unfortunately, it is not guaranteed that~$G_I$ must have a small treewidth due to the edges that connect the targets and destinations of the individual agents.
However, we can bound its treewidth whenever $I$ is a yes-instance.

\begin{claim}\label{claim:yes-instance-tw}
If $I$ is a yes-instance, then the treewidth of $G_{I}$ is at most $O(\ell \cdot \tw)$, where $\tw$ is the treewidth of $G$.
\end{claim}
\begin{proofclaim}
Let $(T, \beta)$ be a tree-decomposition of~$G$ of optimal width~$\tw$.
First, we consider a graph $G'_I$ obtained from $G_I$ by considering only the edges with labels~$\mathsf{copy}$, $\mathsf{communication}$ and $\mathsf{cross}$.
We define its tree-decomposition $(T, \beta')$ by replacing every occurrence of a vertex~$v$ in any bag with its $\ell+1$ copies $v_0, \dots, v_\ell$.
It is easy to see that this is a valid tree-decomposition of~$G'_I$ of width~$O(\ell\cdot \tw)$.

Now, we define a tree-decomposition~$(T,\beta'')$ of~$G_I$ assuming that $I$ is a yes-instance.
By Observation~\ref{obs:yes-instance-paths}, there is a set of vertex-disjoint paths $\{P_a \mid a \in A\}$ connecting the terminals of each agent.
We obtain $\beta''$ from $\beta'$ by adding the vertices $s_0(a)_0$ and $t(a)_\ell$ to every bag intersected by the path~$P_a$ for each $a \in A$.
Observe that for every vertex $v \in G_I$, the set of nodes~$x$ with $v \in \beta(x)$ remains connected and, moreover, we guarantee that $s_0(a)_0$ and $t(a)_\ell$ appear together in some bag, for example, a bag that originally contained only one of the terminals.

Finally, let us bound the width of the tree-decomposition~$(T, \beta'')$.
Since every vertex can lie on at most one path~$P_a$ for some~$a \in A$, we added to~$\beta''(x)$ at most two new vertices for every vertex~$v \in\beta'(x)$.
Therefore,  $\beta''(x)$ contains at most $3\cdot|\beta'(x)|$ vertices and the tree-decomposition $(T, \beta'')$ has width $O(\ell \cdot \tw)$ as promised.
\end{proofclaim}

Now, we show how to encode the existence of a feasible solution into an \MSOtwo sentence.
Formally, we construct \MSO sentence over the signature consisting of a single binary relation symbol $\inc$ that verifies the incidence between a given vertex and edge, and unary relation symbols $\mathsf{agent}$, $\mathsf{copy}$, $\mathsf{cross}$, $\mathsf{communication}$ and $\mathsf{vertex_0}, \dots, \mathsf{vertex_\ell}$ that exactly correspond to the respective edge and vertex labels.
We proceed in two steps.
First, we translate Observation~\ref{obs:yes-instance-paths} into an existential statement about edge and vertex sets in~$G_I$ with special properties expressed only in terms of the labels.
Afterwards, we show how to encode it in \MSOtwo.

\begin{claim}\label{claim:yes-instance-paths-mso}
The instance $I$ is a yes-instance if and only if there exists a set of edges $S$ and sets of vertices $X_0, \dots X_\ell$ in $G_I$ such that
\begin{enumerate}
\item all vertices in $X_i$ are labeled $\mathsf{vertex_i}$ for every $i \in \intzero{\ell}$,\label{cond:paths-mso-1}
\item all edges in $S$ are labeled $\mathsf{copy}$ or $\mathsf{cross}$,\label{cond:paths-mso-2}
\item every vertex in $X_i$ for $i \in [\ell-1]$ is incident to exactly two edges from~$S$ that connect it to vertices in $X_{i-1}$ and~$X_{i+1}$,\label{cond:paths-mso-3}
\item every vertex in $X_0$ and $X_\ell$ is incident to exactly one edge from~$S$,\label{cond:paths-mso-4}
\item every vertex outside of $X_0 \cup \dots \cup X_\ell$ is not incident to any edge from $S$,\label{cond:paths-mso-5}
\item for every edge~$e$ with label~$\mathsf{agent}$, there is a subset $T \subseteq S$ such that the endpoints of~$e$ are incident to exactly one edge in~$T$ and every other vertex is incident to either zero or two edges from~$T$,\label{cond:paths-mso-6}
\item there are no two edges $u_1 v_1, u_2 v_2 \in S$ such that edges $u_1 v_2$ and $u_2 v_1$ both exist and are labeled $\mathsf{copy}$, and\label{cond:paths-mso-7}
\item $X_i$ forms a $d$-connected set with respect to the edges labeled $\mathsf{communication}$ for each $i\in\intzero{\ell}$.\label{cond:paths-mso-8}
\end{enumerate}
\end{claim}
\begin{proofclaim}
First, let us assume that $I$ is a yes-instance and let $\{P_a \mid a \in A\}$ be the set of vertex-disjoint paths guaranteed by Observation~\ref{obs:yes-instance-paths}.
It is straightforward to check that all properties~(\ref{cond:paths-mso-1})--(\ref{cond:paths-mso-8}) are satisfied when $S$ consists of all edges contained in these paths and $X_i$ for each~$i$ consists of all the vertices in the layer~$V_i$ contained in some path~$P_a$.

Now let us assume that there exist sets $S$ and $X_0, \dots, X_\ell$ that satisfy (\ref{cond:paths-mso-1})--(\ref{cond:paths-mso-8}).
The properties~(\ref{cond:paths-mso-1})--(\ref{cond:paths-mso-5}) guarantee that $S$ forms the edge set of a set of vertex-disjoint paths~$\mathcal{P}$ where each path~$P \in \mathcal{P}$ contains exactly one vertex from each layer with endpoints in layers~$V_0$ and~$V_\ell$.
Therefore, the set of paths~$\mathcal{P}$ satisfies condition~(\ref{cond:yes-instance-paths-1}) of Observation~\ref{obs:yes-instance-paths}.
Moreover, the set~$X_i$ contains exactly the vertices from the layer~$V_i$ that lie on some path $P \in \mathcal{P}$.
Property~(\ref{cond:paths-mso-6}) verifies that~$\mathcal{P}$ contains a path~$P_a$ connecting the vertices $s_0(a)_0$ and $t(a)_\ell$ for each $a \in A$, i.e., condition~\ref{cond:yes-instance-paths-2} of Observation~\ref{obs:yes-instance-paths}.
Property~(\ref{cond:paths-mso-7}) is equivalent to condition~(\ref{cond:yes-instance-paths-3}) of Observation~\ref{obs:yes-instance-paths} as it enforces that two agents cannot swap their positions along a single edge.
And finally, property~(\ref{cond:paths-mso-8}) enforces the connectivity requirement of~\MAPFCCShort equivalently to condition~\ref{cond:yes-instance-paths-4} of Observation~\ref{obs:yes-instance-paths}.
\end{proofclaim}

We proceed to construct an \MSOtwo sentence~$\varphi$ equivalent to Claim~\ref{claim:yes-instance-paths-mso}.
Its general form is \[\varphi \coloneqq \exists S, X_0, \dots, X_\ell \; \left( \bigwedge_{i=1}^i \varphi_i(S, X_0,\dots, X_\ell) \right) \,,\]
where each $\varphi_i$ is an \MSOtwo formula with $\ell+2$ free variables $S$, $X_0, \dots, X_\ell$ that encodes the property~$i$ of Claim~\ref{claim:yes-instance-paths-mso}.


In order to simplify the upcoming definitions, we define a predicate $\deg_k(v ,F)$ expressing that vertex~$v$ is incident with exactly~$k$ edges from~$F$ for $k=0,1,2$.
We define:

\begin{align*}
  \deg_0(v,F) &\coloneqq \nexists e \bigl(e \in F \land \inc(v,e)\bigr) \,, \\
  \deg_1(v,F) &\coloneqq \exists e \in F \; \left(\inc(v,e) \land \forall f \in F \Bigl( \bigl(\inc(v, f)\bigr) \rightarrow f = e \Bigr) \right) \,,  \\
  \deg_2(v,F) &\coloneqq \exists e_1, e_2 \in F \; \left(
    \bigl(\inc(v,e_1) \land \inc(v,e_2)\bigr) \land e_1 \neq e_2
    \land \forall f \in F \Bigl( \bigl(\inc(v, f)\bigr) \rightarrow (f = e_1 \lor f = e_2) \Bigr) \right) \,.
\end{align*}


Additionally, we use the predicate `$e = uv$' to express that edge~$e$ joins vertices $u$ and~$v$. 
Formally, it is a syntactic shorthand for $\inc(u,e) \land \inc(v,e) \land u \neq v$.

The encoding of properties (\ref{cond:paths-mso-1})--(\ref{cond:paths-mso-5}) in \MSOtwo is now fairly straightforward albeit technical
\begin{align*}
\varphi_1(S, X_0,\dots, X_\ell) &\coloneqq \bigwedge_{i=0}^\ell \forall v \bigl( v \in X_i \rightarrow \mathsf{vertex_i}(v) \bigr), \\
\varphi_2(S, X_0,\dots, X_\ell) &\coloneqq \forall e \Bigl( e \in S \rightarrow \bigl(\mathsf{copy}(e) \lor \mathsf{cross}(e)\bigr) \Bigr), \\
\varphi_3(S, X_0,\dots, X_\ell) &\coloneqq 
\mathrlap{\bigwedge_{i=1}^{\ell-1}
\left(
\forall v \left(
 v \in X_i \rightarrow
\exists\, u, w, e, f
\left(
\begin{gathered}
\deg_2(v, S) \land u \in X_{i-1} \land\; w \in X_{i+1}\\ \land\; e,f \in S \land e=uv \land f = vw 
\end{gathered}
\right)
\right)\right),}\hspace{40pt}& \\
\varphi_4(S, X_0,\dots, X_\ell) &\coloneqq
\mathrlap{\forall v \bigl( \left(v \in X_0 \lor v \in X_\ell \right) \rightarrow \deg_1(v, S) \bigr),}\hspace{10pt}& \\
\varphi_5(S, X_0,\dots, X_\ell) &\coloneqq \forall v \left( \left(\bigwedge_{i=0}^\ell v \notin X_i \right) \rightarrow \deg_0(v, S) \right).
\end{align*}

The encoding of property~(\ref{cond:paths-mso-6}) is cumbersome but nevertheless still straightforward as

\[
\varphi_6(S, X_0,\dots, X_\ell) \coloneqq \forall e \\
\left(
\mathsf{agent}(e) \rightarrow \exists\, u, v, T
\left(
\begin{gathered}
T \subseteq S \land
e = uv
\land \deg_1(u, T) \land \deg_1(v,T) \\ 
\land\; \forall w \left(
\deg_0(w,T) \lor \deg_2(w,T)
\lor w=v \lor  w=u
\right)
\end{gathered}
\right)
\right)\,.
\]

Property~(\ref{cond:paths-mso-7}) is expressed readily as
\[
\varphi_7(S, X_0,\dots, X_\ell) \coloneqq \nexists u_1, v_1, u_2, v_2, e_1, e_2, f_1, f_2
\left(
\begin{gathered}
	e_1 = u_1 v_1 \land e_2 = u_2 v_2 \land f_1 = u_1 v_2 \land f_2 = u_2 v_1\\ \land\; 
e_1,e_2 \in S 
\land \mathsf{copy}(f_1) \land \mathsf{copy}(f_2)
\end{gathered}
\right)\,.
\]

In order to express the last property, we first need to encode the distance with respect to the edges with the label $\mathsf{communication}$.
Specifically, we define the predicate $\dist_k(u,v)$ that encodes that this distance between $u$ and $v$ is at most~$k$.
We define $\dist_0(u,v) \coloneqq (u = v)$ and we define $\dist_k(u,v)$ for $k > 0$ inductively as
\[
\dist_k(u,v) \coloneqq \dist_{k-1}(u,v) \lor
\exists w, e \Bigl( \dist_{k-1}(u,w) \land \mathsf{inner}(e) \land e = wv \Bigr).
\]

We now wish to define a predicate $connected_k(X)$ verifying that a vertex set~$X$ is $k$-connected with respect to the $\mathsf{communication}$ edges.
However, it is easier to express that $X$ is not $k$-connected.
That happens if and only if we can partition $X$ into two sets $A$ and $B$ such that for any pair of vertices $u \in A$ and $v \in B$, their distance with respect to the $\mathsf{communication}$ edges is strictly more than~$k$.
Hence, we define
\[
connected_k(X) \coloneqq \nexists A,B 
\left(
\begin{gathered}
\forall v \left(v \in X \rightarrow \left(
(v \in A \lor v \in B) \, \land   \neg(v\in A \land v \in B)
\right)\right) \\
\land \; \forall u,v ((u \in A \land v \in B) \rightarrow \neg \dist_k(u,v))
\end{gathered}
\right),
\]
where the second line verifies that $A$, $B$ form a  partition of~$X$.
The definition of $\varphi_8(S, X_0,\dots, X_\ell)$ is immediate
\[\varphi_8(S, X_0,\dots, X_\ell) \coloneqq \bigwedge_{i=0}^\ell connected_d(X_i).\]

\paragraph{Full algorithm.}
The algorithm first computes the treewidth~$\tw$ of~$G$ in $2^{O(\tw^3)} n$ time~\cite{Bodlaender96}.
Afterwards, it constructs the graph~$G_I$ and checks whether its treewidth is within the bound given by Claim~\ref{claim:yes-instance-tw} using the same algorithm as in the first step.
If not, then it immediately outputs a negative answer.
Otherwise, it uses the celebrated Courcelle's theorem~\cite{Courcelle90} to evaluate $\varphi$ on the obtained tree-decomposition of~$G_I$ and outputs its answer.
The correctness follows from Claims~\ref{claim:yes-instance-tw} and~\ref{claim:yes-instance-paths-mso}.
\end{proof}

Interestingly, the $d$-connectivity requirement can be replaced in Theorem~\ref{thm:fpt-tw-d-l} by any property definable in \MSOtwo, e.g., independent or dominating set.
In other words, we can put a wide variety of requirements on the positions of agents in each step, and
we obtain an effective algorithm parameterized by treewidth plus makespan for any such setting.
Note that this closely resembles reconfiguration problems under the parallel token sliding rule~\cite{BBM23,BBDLM21,KS23}.
The input of such problem is a graph~$G$ with two designated vertex sets~$S$ and $T$ of the same size and the question is whether we can move tokens initially placed on~$S$ to~$T$ by moving in each step an arbitrary subset of tokens to their neighbors  where (i) there can be at most one token at a vertex at a time, (ii) tokens cannot swap along an edge, and (iii) all intermediate positions of tokens satisfy a given condition, e.g., being an independent set.
In a timed version of such problem, we are additionally given an upper bound on the makespan~$\ell$ and ask whether the reconfiguration can be carried out in at most~$\ell$ steps.

We can naturally view the agents as tokens moving on a graph from an initial to a final set of positions using the very same reconfiguration rule.
However, the tokens are now labeled, as each single agent has its required target position.
From this point of view, the \MAPFCCShort problem can be seen as a timed labeled $d$-connected set reconfiguration problem.

Recently, \citet{MouawadNRW14} introduced a metatheorem for \MSOtwo-definable timed reconfiguration problems under a different reconfiguration rule.
The machinery of Theorem~\ref{thm:fpt-tw-d-l} can easily be adapted to a metatheorem for \MSOtwo-definable timed labeled reconfiguration problems under the parallel token sliding rule.

Importantly, Theorem~\ref{thm:fpt-tw-d-l} also implies an efficient algorithm parameterized by the number of agents plus the makespan and the communication range in planar graphs and more generally in any class of graphs with bounded local treewidth.
We say that a graph class $\mathcal{G}$ has \emph{bounded local treewidth} if there is a function $f \colon \N \to \N$ such that for every graph $G \in \mathcal{G}$, its every vertex $v$, and every positive integer~$i$ the treewidth of~$G[N^i_G[v]]$ is at most $f(i)$ where $N^i_G[v]$ is the set of all vertices in distance at most~$i$ from~$v$.
Typical examples of graph classes with bounded local treewidth are planar graphs, graphs of bounded genus, and graphs of bounded max degree.

\begin{corollary}\label{cor:fpt-ltw-d-l-k}
	The \MAPFCCShort problem is in \FPT{} parameterized by the number of agents~$k$ plus the makespan~$\ell$ and the communication range~$d$ in any graph class of bounded local treewidth.
\end{corollary}
\begin{proof}
Let $\mathcal{G}$ be a class of bounded local treewidth and let $\langle G, A, s_0 , t, d, \ell\rangle$ be an instance of \MAPFCCShort where $G \in \mathcal{G}$.
Pick an arbitrary agent $a \in A$ and set $G' = G[N^{kd + \ell}_G[s_0(a)]]$, i.e., the subgraph induced by vertices in distance at most $kd+\ell$ from~$s_0(a)$.
Due to the connectivity requirement, all agents must start within distance at most $kd$ from $s_0(a)$ and therefore, they cannot escape~$G'$ within $\ell$ steps.
Moreover, the treewidth of~$G'$ is at most $f(kd+\ell)$ for some function~$f$ depending only on~$\mathcal{G}$.
It remains to invoke the algorithm of Theorem~\ref{thm:fpt-tw-d-l}.
\end{proof}

Since planar graphs have bounded local treewidth~\cite{Eppstein00}, we directly get the following result.
\begin{corollary}\label{cor:fpt-planar-d-l-k}
    The \MAPFCCShort problem is in \FPT{} parameterized by the number of agents~$k$ plus the makespan~$\ell$ and the communication range~$d$ if the input is a planar graph.
\end{corollary}

\section{Conclusion}
In this paper, we initiated the study of the parameterized complexity of the \MAPFCC problem. 
Our first contribution is that even in very restricted settings, where we have a small number of agents that are obliged to move as a cohesive entity, deciding if the instance has a makespan of at most $3$ is infeasible (Theorem~\ref{thm:hard-k-l-d}). We then identify other restrictions on the input that lead to efficient algorithms (Theorems~\ref{thm:fpt-d-k-Delta} through~\ref{thm:fpt-tw-d-l}).
This opens multiple new research directions that can be explored. First and foremost is the question of checking the efficiency of our algorithms in practice. In particular, the \textsf{MSO} encoding we provide for Theorem~\ref{thm:fpt-tw-d-l} is implementable by employing any state-of-the-art \textsf{MSO} solver (e.g.,~\cite{Langer13,BannachB19,Hecher23}). On the other hand, one could also argue that our work provides ample motivation to follow a more heuristic approach. Even in this case, it is worth checking if our exact algorithms can be used as subroutines to improve the effective running time of the state-of-the-art algorithm that is used in practice. Finally, there is the interesting phenomenon exhibited in Corollary~\ref{cor:AMAPFCC-is-hard}, namely that the anonymous version of the problem studied here remains computationally infeasible. It would be interesting to investigate this anonymous version further, both from the theoretical and a practical point of view. 
We consider all of the above important enough to warrant their respective dedicated studies.

\begin{acks}
This work was co-funded by the European Union under the project Robotics and advanced industrial production (reg.\ no.\ CZ.02.01.01/00/22\_008/0004590).
JMK was additionally supported by the Grant Agency of the Czech Technical University in Prague, grant No.\ SGS23/205/OHK3/3T/18. FF and NM acknowledge the support by the CTU Global postdoc fellowship program.
DK, JMK, and MO acknowledge the support of the Czech Science Foundation Grant No.\ 22-19557S.

A preliminary version of our work was presented at AAAI'25~\cite{FKKMM25}.
\end{acks}

\bibliographystyle{ACM-Reference-Format}
\bibliography{sample-base}


\begin{thebibliography}{50}


\ifx \showCODEN    \undefined \def \showCODEN     #1{\unskip}     \fi
\ifx \showDOI      \undefined \def \showDOI       #1{#1}\fi
\ifx \showISBNx    \undefined \def \showISBNx     #1{\unskip}     \fi
\ifx \showISBNxiii \undefined \def \showISBNxiii  #1{\unskip}     \fi
\ifx \showISSN     \undefined \def \showISSN      #1{\unskip}     \fi
\ifx \showLCCN     \undefined \def \showLCCN      #1{\unskip}     \fi
\ifx \shownote     \undefined \def \shownote      #1{#1}          \fi
\ifx \showarticletitle \undefined \def \showarticletitle #1{#1}   \fi
\ifx \showURL      \undefined \def \showURL       {\relax}        \fi
\providecommand\bibfield[2]{#2}
\providecommand\bibinfo[2]{#2}
\providecommand\natexlab[1]{#1}
\providecommand\showeprint[2][]{arXiv:#2}

\bibitem[Ali and Yakovlev(2024)]%
        {AY24}
\bibfield{author}{\bibinfo{person}{Zain~Alabedeen Ali} {and}
  \bibinfo{person}{Konstantin~S. Yakovlev}.} \bibinfo{year}{2024}\natexlab{}.
\newblock \showarticletitle{Improved Anonymous Multi-Agent Path Finding
  Algorithm}. In \bibinfo{booktitle}{\emph{Thirty-Eighth {AAAI} Conference on
  Artificial Intelligence, {AAAI} 2024, Thirty-Sixth Conference on Innovative
  Applications of Artificial Intelligence, {IAAI} 2024, Fourteenth Symposium on
  Educational Advances in Artificial Intelligence, {EAAI} 2014, February 20-27,
  2024, Vancouver, Canada}}, \bibfield{editor}{\bibinfo{person}{Michael~J.
  Wooldridge}, \bibinfo{person}{Jennifer~G. Dy}, {and} \bibinfo{person}{Sriraam
  Natarajan}} (Eds.). \bibinfo{publisher}{{AAAI} Press},
  \bibinfo{pages}{17291--17298}.
\newblock
\urldef\tempurl%
\url{https://doi.org/10.1609/aaai.v38i16.29676}
\showURL{%
\tempurl}


\bibitem[Amigoni et~al\mbox{.}(2017)]%
        {amigoni2017multirobot}
\bibfield{author}{\bibinfo{person}{Francesco Amigoni}, \bibinfo{person}{Jacopo
  Banfi}, {and} \bibinfo{person}{Nicola Basilico}.}
  \bibinfo{year}{2017}\natexlab{}.
\newblock \showarticletitle{Multirobot exploration of communication-restricted
  environments: A survey}.
\newblock \bibinfo{journal}{\emph{IEEE Intelligent Systems}}
  \bibinfo{volume}{32}, \bibinfo{number}{6} (\bibinfo{year}{2017}),
  \bibinfo{pages}{48--57}.
\newblock


\bibitem[Arnborg et~al\mbox{.}(1987)]%
        {ACP87}
\bibfield{author}{\bibinfo{person}{Stefan Arnborg}, \bibinfo{person}{Derek~G.
  Corneil}, {and} \bibinfo{person}{Andrzej Proskurowski}.}
  \bibinfo{year}{1987}\natexlab{}.
\newblock \showarticletitle{Complexity of Finding Embeddings in a k-Tree}.
\newblock \bibinfo{journal}{\emph{SIAM Journal on Algebraic Discrete Methods}}
  \bibinfo{volume}{8}, \bibinfo{number}{2} (\bibinfo{year}{1987}),
  \bibinfo{pages}{277--284}.
\newblock
\urldef\tempurl%
\url{https://doi.org/10.1137/0608024}
\showDOI{\tempurl}


\bibitem[Bannach and Berndt(2019)]%
        {BannachB19}
\bibfield{author}{\bibinfo{person}{Max Bannach} {and}
  \bibinfo{person}{Sebastian Berndt}.} \bibinfo{year}{2019}\natexlab{}.
\newblock \showarticletitle{Practical Access to Dynamic Programming on Tree
  Decompositions}.
\newblock \bibinfo{journal}{\emph{Algorithms}} \bibinfo{volume}{12},
  \bibinfo{number}{8} (\bibinfo{year}{2019}), \bibinfo{pages}{172}.
\newblock
\urldef\tempurl%
\url{https://doi.org/10.3390/A12080172}
\showDOI{\tempurl}


\bibitem[Barer et~al\mbox{.}(2014)]%
        {barer2014suboptimal}
\bibfield{author}{\bibinfo{person}{Max Barer}, \bibinfo{person}{Guni Sharon},
  \bibinfo{person}{Roni Stern}, {and} \bibinfo{person}{Ariel Felner}.}
  \bibinfo{year}{2014}\natexlab{}.
\newblock \showarticletitle{Suboptimal variants of the conflict-based search
  algorithm for the multi-agent pathfinding problem}. In
  \bibinfo{booktitle}{\emph{International Symposium on Combinatorial Search}}.
  \bibinfo{pages}{19--27}.
\newblock


\bibitem[Bart{\'{a}}k et~al\mbox{.}(2017)]%
        {BartakZSBS17}
\bibfield{author}{\bibinfo{person}{Roman Bart{\'{a}}k},
  \bibinfo{person}{Neng{-}Fa Zhou}, \bibinfo{person}{Roni Stern},
  \bibinfo{person}{Eli Boyarski}, {and} \bibinfo{person}{Pavel Surynek}.}
  \bibinfo{year}{2017}\natexlab{}.
\newblock \showarticletitle{Modeling and Solving the Multi-agent Pathfinding
  Problem in Picat}. In \bibinfo{booktitle}{\emph{29th IEEE International
  Conference on Tools with Artificial Intelligence}}.
  \bibinfo{pages}{959--966}.
\newblock
\urldef\tempurl%
\url{https://doi.org/10.1109/ICTAI.2017.00147}
\showDOI{\tempurl}


\bibitem[Bartier et~al\mbox{.}(2021)]%
        {BBDLM21}
\bibfield{author}{\bibinfo{person}{Valentin Bartier}, \bibinfo{person}{Nicolas
  Bousquet}, \bibinfo{person}{Cl{\'{e}}ment Dallard}, \bibinfo{person}{Kyle
  Lomer}, {and} \bibinfo{person}{Amer~E. Mouawad}.}
  \bibinfo{year}{2021}\natexlab{}.
\newblock \showarticletitle{On Girth and the Parameterized Complexity of Token
  Sliding and Token Jumping}.
\newblock \bibinfo{journal}{\emph{Algorithmica}} \bibinfo{volume}{83},
  \bibinfo{number}{9} (\bibinfo{year}{2021}), \bibinfo{pages}{2914--2951}.
\newblock
\urldef\tempurl%
\url{https://doi.org/10.1007/S00453-021-00848-1}
\showDOI{\tempurl}


\bibitem[Bartier et~al\mbox{.}(2023)]%
        {BBM23}
\bibfield{author}{\bibinfo{person}{Valentin Bartier}, \bibinfo{person}{Nicolas
  Bousquet}, {and} \bibinfo{person}{Amer~E. Mouawad}.}
  \bibinfo{year}{2023}\natexlab{}.
\newblock \showarticletitle{Galactic token sliding}.
\newblock \bibinfo{journal}{\emph{J. Comput. Syst. Sci.}}
  \bibinfo{volume}{136} (\bibinfo{year}{2023}), \bibinfo{pages}{220--248}.
\newblock
\urldef\tempurl%
\url{https://doi.org/10.1016/J.JCSS.2023.03.008}
\showDOI{\tempurl}


\bibitem[Bodlaender(1996)]%
        {Bodlaender96}
\bibfield{author}{\bibinfo{person}{Hans~L. Bodlaender}.}
  \bibinfo{year}{1996}\natexlab{}.
\newblock \showarticletitle{A Linear-Time Algorithm for Finding
  Tree-Decompositions of Small Treewidth}.
\newblock \bibinfo{journal}{\emph{SIAM J. Comput.}} \bibinfo{volume}{25},
  \bibinfo{number}{6} (\bibinfo{year}{1996}), \bibinfo{pages}{1305--1317}.
\newblock
\urldef\tempurl%
\url{https://doi.org/10.1137/S0097539793251219}
\showDOI{\tempurl}


\bibitem[Boyarski et~al\mbox{.}(2015)]%
        {BoyarskiFSSBTS15}
\bibfield{author}{\bibinfo{person}{Eli Boyarski}, \bibinfo{person}{Ariel
  Felner}, \bibinfo{person}{Roni Stern}, \bibinfo{person}{Guni Sharon},
  \bibinfo{person}{Oded Betzalel}, \bibinfo{person}{David Tolpin}, {and}
  \bibinfo{person}{Solomon~Eyal Shimony}.} \bibinfo{year}{2015}\natexlab{}.
\newblock \showarticletitle{{ICBS:} The Improved Conflict-Based Search
  Algorithm for Multi-Agent Pathfinding}. In \bibinfo{booktitle}{\emph{8th
  Annual Symposium on Combinatorial Search}}. \bibinfo{pages}{223--225}.
\newblock
\urldef\tempurl%
\url{https://doi.org/10.1609/SOCS.V6I1.18343}
\showDOI{\tempurl}


\bibitem[Calviac et~al\mbox{.}(2023)]%
        {calviac2023improved}
\bibfield{author}{\bibinfo{person}{Isse{\"\i}nie Calviac},
  \bibinfo{person}{Ocan Sankur}, {and} \bibinfo{person}{Fran{\c{c}}ois
  Schwarzentruber}.} \bibinfo{year}{2023}\natexlab{}.
\newblock \showarticletitle{Improved Complexity Results and an Efficient
  Solution for Connected Multi-Agent Path Finding}. In
  \bibinfo{booktitle}{\emph{22nd International Conference on Autonomous Agents
  and Multiagent Systems}}. \bibinfo{pages}{1--9}.
\newblock


\bibitem[Courcelle(1990)]%
        {Courcelle90}
\bibfield{author}{\bibinfo{person}{Bruno Courcelle}.}
  \bibinfo{year}{1990}\natexlab{}.
\newblock \showarticletitle{The Monadic Second-Order Logic of Graphs. I.
  Recognizable Sets of Finite Graphs}.
\newblock \bibinfo{journal}{\emph{Information and Computation}}
  \bibinfo{volume}{85}, \bibinfo{number}{1} (\bibinfo{year}{1990}),
  \bibinfo{pages}{12--75}.
\newblock
\urldef\tempurl%
\url{https://doi.org/10.1016/0890-5401(90)90043-H}
\showDOI{\tempurl}


\bibitem[Cygan et~al\mbox{.}(2015)]%
        {CyganFKLMPPS15}
\bibfield{author}{\bibinfo{person}{Marek Cygan}, \bibinfo{person}{Fedor~V.
  Fomin}, \bibinfo{person}{\L{}ukasz Kowalik}, \bibinfo{person}{Daniel
  Lokshtanov}, \bibinfo{person}{D{\'{a}}niel Marx}, \bibinfo{person}{Marcin
  Pilipczuk}, \bibinfo{person}{Micha\l{} Pilipczuk}, {and}
  \bibinfo{person}{Saket Saurabh}.} \bibinfo{year}{2015}\natexlab{}.
\newblock \bibinfo{booktitle}{\emph{Parameterized Algorithms}}.
\newblock \bibinfo{publisher}{Springer}.
\newblock
\showISBNx{978-3-319-21274-6}
\urldef\tempurl%
\url{https://doi.org/10.1007/978-3-319-21275-3}
\showDOI{\tempurl}


\bibitem[Downey and Fellows(2012)]%
        {downey2012parameterized}
\bibfield{author}{\bibinfo{person}{Rodney~G Downey} {and}
  \bibinfo{person}{Michael~Ralph Fellows}.} \bibinfo{year}{2012}\natexlab{}.
\newblock \bibinfo{booktitle}{\emph{Parameterized complexity}}.
\newblock \bibinfo{publisher}{Springer Science \& Business Media}.
\newblock


\bibitem[Eiben et~al\mbox{.}(2023)]%
        {EibenGK23}
\bibfield{author}{\bibinfo{person}{Eduard Eiben}, \bibinfo{person}{Robert
  Ganian}, {and} \bibinfo{person}{Iyad Kanj}.} \bibinfo{year}{2023}\natexlab{}.
\newblock \showarticletitle{The Parameterized Complexity of Coordinated Motion
  Planning}. In \bibinfo{booktitle}{\emph{39th International Symposium on
  Computational Geometry}}. \bibinfo{pages}{28:1--28:16}.
\newblock
\urldef\tempurl%
\url{https://doi.org/10.4230/LIPIcs.SoCG.2023.28}
\showDOI{\tempurl}


\bibitem[Eppstein(2000)]%
        {Eppstein00}
\bibfield{author}{\bibinfo{person}{David Eppstein}.}
  \bibinfo{year}{2000}\natexlab{}.
\newblock \showarticletitle{Diameter and Treewidth in Minor-Closed Graph
  Families}.
\newblock \bibinfo{journal}{\emph{Algorithmica}} \bibinfo{volume}{27},
  \bibinfo{number}{3} (\bibinfo{year}{2000}), \bibinfo{pages}{275--291}.
\newblock
\urldef\tempurl%
\url{https://doi.org/10.1007/S004530010020}
\showDOI{\tempurl}


\bibitem[Felner et~al\mbox{.}(2017)]%
        {FelnerSSBGSSWS17}
\bibfield{author}{\bibinfo{person}{Ariel Felner}, \bibinfo{person}{Roni Stern},
  \bibinfo{person}{Solomon~Eyal Shimony}, \bibinfo{person}{Eli Boyarski},
  \bibinfo{person}{Meir Goldenberg}, \bibinfo{person}{Guni Sharon},
  \bibinfo{person}{Nathan~R. Sturtevant}, \bibinfo{person}{Glenn Wagner}, {and}
  \bibinfo{person}{Pavel Surynek}.} \bibinfo{year}{2017}\natexlab{}.
\newblock \showarticletitle{Search-Based Optimal Solvers for the Multi-Agent
  Pathfinding Problem: Summary and Challenges}. In
  \bibinfo{booktitle}{\emph{10th International Symposium on Combinatorial
  Search}}. \bibinfo{pages}{29--37}.
\newblock
\urldef\tempurl%
\url{https://doi.org/10.1609/SOCS.V8I1.18423}
\showDOI{\tempurl}


\bibitem[Fioravantes et~al\mbox{.}(2024)]%
        {FKKMO24}
\bibfield{author}{\bibinfo{person}{Foivos Fioravantes}, \bibinfo{person}{Dusan
  Knop}, \bibinfo{person}{Jan~Maty{\'{a}}s Kristan}, \bibinfo{person}{Nikolaos
  Melissinos}, {and} \bibinfo{person}{Michal Opler}.}
  \bibinfo{year}{2024}\natexlab{}.
\newblock \showarticletitle{Exact Algorithms and Lowerbounds for Multiagent
  Path Finding: Power of Treelike Topology}. In
  \bibinfo{booktitle}{\emph{Thirty-Eighth AAAI Conference on Artificial
  Intelligence}}. \bibinfo{pages}{17380--17388}.
\newblock
\urldef\tempurl%
\url{https://doi.org/10.1609/AAAI.V38I16.29686}
\showDOI{\tempurl}


\bibitem[Fioravantes et~al\mbox{.}(2025)]%
        {FKKMM25}
\bibfield{author}{\bibinfo{person}{Foivos Fioravantes}, \bibinfo{person}{Dusan
  Knop}, \bibinfo{person}{Jan~Maty{\'{a}}s Kristan}, \bibinfo{person}{Nikolaos
  Melissinos}, {and} \bibinfo{person}{Michal Opler}.}
  \bibinfo{year}{2025}\natexlab{}.
\newblock \showarticletitle{Exact Algorithms for Multiagent Path Finding with
  Communication Constraints on Tree-Like Structures}. In
  \bibinfo{booktitle}{\emph{AAAI-25, Sponsored by the Association for the
  Advancement of Artificial Intelligence, February 25 - March 4, 2025,
  Philadelphia, PA, {USA}}}. \bibinfo{publisher}{{AAAI} Press},
  \bibinfo{pages}{23177--23185}.
\newblock
\urldef\tempurl%
\url{https://doi.org/10.1609/AAAI.V39I22.34483}
\showDOI{\tempurl}


\bibitem[Flum and Grohe(2006)]%
        {FlumG06}
\bibfield{author}{\bibinfo{person}{J{\"{o}}rg Flum} {and}
  \bibinfo{person}{Martin Grohe}.} \bibinfo{year}{2006}\natexlab{}.
\newblock \bibinfo{booktitle}{\emph{Parameterized Complexity Theory}}.
\newblock \bibinfo{publisher}{Springer}.
\newblock
\showISBNx{978-3-540-29952-3}
\urldef\tempurl%
\url{https://doi.org/10.1007/3-540-29953-X}
\showDOI{\tempurl}


\bibitem[Ganian et~al\mbox{.}(2021)]%
        {GOR21}
\bibfield{author}{\bibinfo{person}{Robert Ganian}, \bibinfo{person}{Sebastian
  Ordyniak}, {and} \bibinfo{person}{M.~S. Ramanujan}.}
  \bibinfo{year}{2021}\natexlab{}.
\newblock \showarticletitle{On Structural Parameterizations of the Edge
  Disjoint Paths Problem}.
\newblock \bibinfo{journal}{\emph{Algorithmica}} \bibinfo{volume}{83},
  \bibinfo{number}{6} (\bibinfo{year}{2021}), \bibinfo{pages}{1605--1637}.
\newblock
\urldef\tempurl%
\url{https://doi.org/10.1007/S00453-020-00795-3}
\showDOI{\tempurl}


\bibitem[Goldreich(2011)]%
        {Goldreich2011}
\bibfield{author}{\bibinfo{person}{Oded Goldreich}.}
  \bibinfo{year}{2011}\natexlab{}.
\newblock \bibinfo{booktitle}{\emph{Finding the Shortest Move-Sequence in the
  Graph-Generalized 15-Puzzle Is NP-Hard}}.
\newblock \bibinfo{publisher}{Springer}, \bibinfo{address}{Berlin, Heidelberg},
  \bibinfo{pages}{1--5}.
\newblock
\showISBNx{978-3-642-22670-0}
\urldef\tempurl%
\url{https://doi.org/10.1007/978-3-642-22670-0_1}
\showDOI{\tempurl}


\bibitem[Guillemot and Marx(2014)]%
        {GuillemotM14}
\bibfield{author}{\bibinfo{person}{Sylvain Guillemot} {and}
  \bibinfo{person}{D{\'{a}}niel Marx}.} \bibinfo{year}{2014}\natexlab{}.
\newblock \showarticletitle{Finding small patterns in permutations in linear
  time}. In \bibinfo{booktitle}{\emph{Proceedings of the Twenty-Fifth Annual
  Symposium on Discrete Algorithms, {SODA} 2014}}. \bibinfo{pages}{82--101}.
\newblock
\urldef\tempurl%
\url{https://doi.org/10.1137/1.9781611973402.7}
\showDOI{\tempurl}


\bibitem[Gupta et~al\mbox{.}(2017)]%
        {guptaEK2017}
\bibfield{author}{\bibinfo{person}{Jayesh~K. Gupta}, \bibinfo{person}{Maxim
  Egorov}, {and} \bibinfo{person}{Mykel Kochenderfer}.}
  \bibinfo{year}{2017}\natexlab{}.
\newblock \showarticletitle{Cooperative Multi-agent Control Using Deep
  Reinforcement Learning}. In \bibinfo{booktitle}{\emph{Autonomous Agents and
  Multiagent Systems: AAMAS 2017 Workshops}}. \bibinfo{pages}{66--83}.
\newblock


\bibitem[Hearn and Demaine(2005)]%
        {hearn2005pspace}
\bibfield{author}{\bibinfo{person}{Robert~A Hearn} {and}
  \bibinfo{person}{Erik~D Demaine}.} \bibinfo{year}{2005}\natexlab{}.
\newblock \showarticletitle{PSPACE-completeness of sliding-block puzzles and
  other problems through the nondeterministic constraint logic model of
  computation}.
\newblock \bibinfo{journal}{\emph{Theoretical Computer Science}}
  \bibinfo{volume}{343}, \bibinfo{number}{1-2} (\bibinfo{year}{2005}),
  \bibinfo{pages}{72--96}.
\newblock


\bibitem[Hecher(2023)]%
        {Hecher23}
\bibfield{author}{\bibinfo{person}{Markus Hecher}.}
  \bibinfo{year}{2023}\natexlab{}.
\newblock \showarticletitle{Advanced tools and methods for treewidth-based
  problem solving}.
\newblock \bibinfo{journal}{\emph{IT-Information Technology}}
  \bibinfo{volume}{65}, \bibinfo{number}{1-2} (\bibinfo{year}{2023}),
  \bibinfo{pages}{65--73}.
\newblock
\urldef\tempurl%
\url{https://doi.org/10.1515/itit-2023-0004}
\showDOI{\tempurl}


\bibitem[Hollinger and Singh(2012)]%
        {hollinger2012multirobot}
\bibfield{author}{\bibinfo{person}{Geoffrey~A Hollinger} {and}
  \bibinfo{person}{Sanjiv Singh}.} \bibinfo{year}{2012}\natexlab{}.
\newblock \showarticletitle{Multirobot coordination with periodic connectivity:
  Theory and experiments}.
\newblock \bibinfo{journal}{\emph{IEEE Transactions on Robotics}}
  \bibinfo{volume}{28}, \bibinfo{number}{4} (\bibinfo{year}{2012}),
  \bibinfo{pages}{967--973}.
\newblock


\bibitem[Kloks(1994)]%
        {Kloks94}
\bibfield{author}{\bibinfo{person}{Ton Kloks}.}
  \bibinfo{year}{1994}\natexlab{}.
\newblock \bibinfo{booktitle}{\emph{Treewidth, Computations and
  Approximations}}. \bibinfo{series}{Lecture Notes in Computer Science},
  Vol.~\bibinfo{volume}{842}.
\newblock \bibinfo{publisher}{Springer}.
\newblock
\showISBNx{3-540-58356-4}
\urldef\tempurl%
\url{https://doi.org/10.1007/BFb0045375}
\showDOI{\tempurl}


\bibitem[Korf(1990)]%
        {korf1990real}
\bibfield{author}{\bibinfo{person}{Richard~E Korf}.}
  \bibinfo{year}{1990}\natexlab{}.
\newblock \showarticletitle{Real-time heuristic search}.
\newblock \bibinfo{journal}{\emph{Artificial intelligence}}
  \bibinfo{volume}{42}, \bibinfo{number}{2-3} (\bibinfo{year}{1990}),
  \bibinfo{pages}{189--211}.
\newblock


\bibitem[Korhonen and Lokshtanov(2023)]%
        {KorhonenL23}
\bibfield{author}{\bibinfo{person}{Tuukka Korhonen} {and}
  \bibinfo{person}{Daniel Lokshtanov}.} \bibinfo{year}{2023}\natexlab{}.
\newblock \showarticletitle{An Improved Parameterized Algorithm for Treewidth}.
  In \bibinfo{booktitle}{\emph{Proceedings of the 55th Annual {ACM} Symposium
  on Theory of Computing, {STOC} 2023, Orlando, FL, USA, June 20-23, 2023}},
  \bibfield{editor}{\bibinfo{person}{Barna Saha} {and}
  \bibinfo{person}{Rocco~A. Servedio}} (Eds.). \bibinfo{publisher}{{ACM}},
  \bibinfo{pages}{528--541}.
\newblock
\urldef\tempurl%
\url{https://doi.org/10.1145/3564246.3585245}
\showDOI{\tempurl}


\bibitem[Kornhauser et~al\mbox{.}(1984)]%
        {KornhauserMS84}
\bibfield{author}{\bibinfo{person}{Daniel Kornhauser}, \bibinfo{person}{Gary~L.
  Miller}, {and} \bibinfo{person}{Paul~G. Spirakis}.}
  \bibinfo{year}{1984}\natexlab{}.
\newblock \showarticletitle{Coordinating Pebble Motion on Graphs, the Diameter
  of Permutation Groups, and Applications}. In \bibinfo{booktitle}{\emph{25th
  Annual Symposium on Foundations of Computer Science}}.
  \bibinfo{pages}{241--250}.
\newblock
\urldef\tempurl%
\url{https://doi.org/10.1109/SFCS.1984.715921}
\showDOI{\tempurl}


\bibitem[Kristan and Svoboda(2023)]%
        {KS23}
\bibfield{author}{\bibinfo{person}{Jan~Maty{\'{a}}s Kristan} {and}
  \bibinfo{person}{Jakub Svoboda}.} \bibinfo{year}{2023}\natexlab{}.
\newblock \showarticletitle{Shortest Dominating Set Reconfiguration Under Token
  Sliding}. In \bibinfo{booktitle}{\emph{Fundamentals of Computation Theory -
  24th International Symposium, {FCT} 2023, Trier, Germany, September 18-21,
  2023, Proceedings}} \emph{(\bibinfo{series}{Lecture Notes in Computer
  Science}, Vol.~\bibinfo{volume}{14292})},
  \bibfield{editor}{\bibinfo{person}{Henning Fernau} {and}
  \bibinfo{person}{Klaus Jansen}} (Eds.). \bibinfo{publisher}{Springer},
  \bibinfo{pages}{333--347}.
\newblock
\urldef\tempurl%
\url{https://doi.org/10.1007/978-3-031-43587-4\_24}
\showDOI{\tempurl}


\bibitem[Langer(2013)]%
        {Langer13}
\bibfield{author}{\bibinfo{person}{Alexander Langer}.}
  \bibinfo{year}{2013}\natexlab{}.
\newblock \emph{\bibinfo{title}{Fast algorithms for decomposable graphs}}.
\newblock \bibinfo{thesistype}{Ph.\,D. Dissertation}. \bibinfo{school}{{RWTH}
  Aachen University}.
\newblock
\urldef\tempurl%
\url{http://darwin.bth.rwth-aachen.de/opus3/volltexte/2013/4531}
\showURL{%
\tempurl}


\bibitem[Mouawad et~al\mbox{.}(2014)]%
        {MouawadNRW14}
\bibfield{author}{\bibinfo{person}{Amer~E. Mouawad}, \bibinfo{person}{Naomi
  Nishimura}, \bibinfo{person}{Venkatesh Raman}, {and} \bibinfo{person}{Marcin
  Wrochna}.} \bibinfo{year}{2014}\natexlab{}.
\newblock \showarticletitle{Reconfiguration over Tree Decompositions}. In
  \bibinfo{booktitle}{\emph{9th International Symposium on Parameterized and
  Exact Computation}}. \bibinfo{pages}{246--257}.
\newblock
\urldef\tempurl%
\url{https://doi.org/10.1007/978-3-319-13524-3\_21}
\showDOI{\tempurl}


\bibitem[Niedermeier(2006)]%
        {Niedermeier06}
\bibfield{author}{\bibinfo{person}{Rolf Niedermeier}.}
  \bibinfo{year}{2006}\natexlab{}.
\newblock \bibinfo{booktitle}{\emph{Invitation to Fixed-Parameter Algorithms}}.
\newblock \bibinfo{publisher}{Oxford University Press}.
\newblock
\showISBNx{9780198566076}
\urldef\tempurl%
\url{https://doi.org/10.1093/ACPROF:OSO/9780198566076.001.0001}
\showDOI{\tempurl}


\bibitem[Pietrzak(2003)]%
        {Pietrzak03}
\bibfield{author}{\bibinfo{person}{Krzysztof Pietrzak}.}
  \bibinfo{year}{2003}\natexlab{}.
\newblock \showarticletitle{On the parameterized complexity of the fixed
  alphabet shortest common supersequence and longest common subsequence
  problems}.
\newblock \bibinfo{journal}{\emph{J. Comput. System Sci.}}
  \bibinfo{volume}{67} (\bibinfo{date}{12} \bibinfo{year}{2003}),
  \bibinfo{pages}{757--771}.
\newblock
\urldef\tempurl%
\url{https://doi.org/10.1016/S0022-0000(03)00078-3}
\showDOI{\tempurl}


\bibitem[Sharon et~al\mbox{.}(2015)]%
        {sharon2015conflict}
\bibfield{author}{\bibinfo{person}{Guni Sharon}, \bibinfo{person}{Roni Stern},
  \bibinfo{person}{Ariel Felner}, {and} \bibinfo{person}{Nathan~R Sturtevant}.}
  \bibinfo{year}{2015}\natexlab{}.
\newblock \showarticletitle{Conflict-based search for optimal multi-agent
  pathfinding}.
\newblock \bibinfo{journal}{\emph{Artificial intelligence}}
  \bibinfo{volume}{219} (\bibinfo{year}{2015}), \bibinfo{pages}{40--66}.
\newblock


\bibitem[Sharon et~al\mbox{.}(2013)]%
        {sharon2013increasing}
\bibfield{author}{\bibinfo{person}{Guni Sharon}, \bibinfo{person}{Roni Stern},
  \bibinfo{person}{Meir Goldenberg}, {and} \bibinfo{person}{Ariel Felner}.}
  \bibinfo{year}{2013}\natexlab{}.
\newblock \showarticletitle{The increasing cost tree search for optimal
  multi-agent pathfinding}.
\newblock \bibinfo{journal}{\emph{Artificial intelligence}}
  \bibinfo{volume}{195} (\bibinfo{year}{2013}), \bibinfo{pages}{470--495}.
\newblock


\bibitem[Silver(2005)]%
        {silver2005cooperative}
\bibfield{author}{\bibinfo{person}{David Silver}.}
  \bibinfo{year}{2005}\natexlab{}.
\newblock \showarticletitle{Cooperative pathfinding}. In
  \bibinfo{booktitle}{\emph{AAAI Conference on Artificial Intelligence and
  Interactive Digital Entertainment}}. \bibinfo{pages}{117--122}.
\newblock


\bibitem[Snape et~al\mbox{.}(2012)]%
        {SnapeGBLM12}
\bibfield{author}{\bibinfo{person}{Jamie Snape}, \bibinfo{person}{Stephen~J.
  Guy}, \bibinfo{person}{Jur van~den Berg}, \bibinfo{person}{Ming~C. Lin},
  {and} \bibinfo{person}{Dinesh Manocha}.} \bibinfo{year}{2012}\natexlab{}.
\newblock \showarticletitle{Reciprocal Collision Avoidance and Multi-Agent
  Navigation for Video Games}. In \bibinfo{booktitle}{\emph{AAAI Workshop on
  Multiagent Pathfinding (MAPF@AAAI 2012)}}.
\newblock


\bibitem[Stern et~al\mbox{.}(2019)]%
        {SternSFK0WLA0KB19manysurvey}
\bibfield{author}{\bibinfo{person}{Roni Stern}, \bibinfo{person}{Nathan~R.
  Sturtevant}, \bibinfo{person}{Ariel Felner}, \bibinfo{person}{Sven Koenig},
  \bibinfo{person}{Hang Ma}, \bibinfo{person}{Thayne~T. Walker},
  \bibinfo{person}{Jiaoyang Li}, \bibinfo{person}{Dor Atzmon},
  \bibinfo{person}{Liron Cohen}, \bibinfo{person}{T.~K.~Satish Kumar},
  \bibinfo{person}{Roman Bart{\'{a}}k}, {and} \bibinfo{person}{Eli Boyarski}.}
  \bibinfo{year}{2019}\natexlab{}.
\newblock \showarticletitle{Multi-Agent Pathfinding: Definitions, Variants, and
  Benchmarks}. In \bibinfo{booktitle}{\emph{12th International Symposium on
  Combinatorial Search}}. \bibinfo{pages}{151--158}.
\newblock
\urldef\tempurl%
\url{https://doi.org/10.1609/SOCS.V10I1.18510}
\showDOI{\tempurl}


\bibitem[Surynek(2010)]%
        {surynek2010optimization}
\bibfield{author}{\bibinfo{person}{Pavel Surynek}.}
  \bibinfo{year}{2010}\natexlab{}.
\newblock \showarticletitle{An optimization variant of multi-robot path
  planning is intractable}. In \bibinfo{booktitle}{\emph{AAAI Conference on
  Artificial Intelligence}}. \bibinfo{pages}{1261--1263}.
\newblock


\bibitem[Surynek(2022)]%
        {surynek2022survey}
\bibfield{author}{\bibinfo{person}{Pavel Surynek}.}
  \bibinfo{year}{2022}\natexlab{}.
\newblock \showarticletitle{Problem Compilation for Multi-Agent Path Finding: a
  Survey}. In \bibinfo{booktitle}{\emph{31st International Joint Conference on
  Artificial Intelligence}}. \bibinfo{pages}{5615--5622}.
\newblock
\urldef\tempurl%
\url{https://doi.org/10.24963/ijcai.2022/783}
\showDOI{\tempurl}
\newblock
\shownote{Survey Track}.


\bibitem[Surynek et~al\mbox{.}(2022)]%
        {SurynekSBF22}
\bibfield{author}{\bibinfo{person}{Pavel Surynek}, \bibinfo{person}{Roni
  Stern}, \bibinfo{person}{Eli Boyarski}, {and} \bibinfo{person}{Ariel
  Felner}.} \bibinfo{year}{2022}\natexlab{}.
\newblock \showarticletitle{Migrating Techniques from Search-based Multi-Agent
  Path Finding Solvers to SAT-based Approach}.
\newblock \bibinfo{journal}{\emph{Journal of Artificial Intelligence Research}}
   \bibinfo{volume}{73} (\bibinfo{year}{2022}), \bibinfo{pages}{553--618}.
\newblock
\urldef\tempurl%
\url{https://doi.org/10.1613/JAIR.1.13318}
\showDOI{\tempurl}


\bibitem[Tateo et~al\mbox{.}(2018)]%
        {tateo2018multiagent}
\bibfield{author}{\bibinfo{person}{Davide Tateo}, \bibinfo{person}{Jacopo
  Banfi}, \bibinfo{person}{Alessandro Riva}, \bibinfo{person}{Francesco
  Amigoni}, {and} \bibinfo{person}{Andrea Bonarini}.}
  \bibinfo{year}{2018}\natexlab{}.
\newblock \showarticletitle{Multiagent connected path planning:
  Pspace-completeness and how to deal with it}. In
  \bibinfo{booktitle}{\emph{AAAI Conference on Artificial Intelligence}}.
\newblock


\bibitem[Wilson(1974)]%
        {Wilson74}
\bibfield{author}{\bibinfo{person}{Richard~M. Wilson}.}
  \bibinfo{year}{1974}\natexlab{}.
\newblock \showarticletitle{Graph puzzles, homotopy, and the alternating
  group}.
\newblock \bibinfo{journal}{\emph{Journal of Combinatorial Theory, Series B}}
  \bibinfo{volume}{16} (\bibinfo{year}{1974}), \bibinfo{pages}{86--96}.
\newblock
\urldef\tempurl%
\url{https://doi.org/10.1016/0095-8956(74)90098-7}
\showDOI{\tempurl}


\bibitem[Yu and LaValle(2012)]%
        {YL12}
\bibfield{author}{\bibinfo{person}{Jingjin Yu} {and} \bibinfo{person}{Steven~M.
  LaValle}.} \bibinfo{year}{2012}\natexlab{}.
\newblock \showarticletitle{Multi-agent Path Planning and Network Flow}. In
  \bibinfo{booktitle}{\emph{Algorithmic Foundations of Robotics {X} -
  Proceedings of the Tenth Workshop on the Algorithmic Foundations of Robotics,
  {WAFR} 2012, MIT, Cambridge, Massachusetts, USA, June 13-15 2012}}
  \emph{(\bibinfo{series}{Springer Tracts in Advanced Robotics},
  Vol.~\bibinfo{volume}{86})}, \bibfield{editor}{\bibinfo{person}{Emilio
  Frazzoli}, \bibinfo{person}{Tom{\'{a}}s Lozano{-}P{\'{e}}rez},
  \bibinfo{person}{Nicholas Roy}, {and} \bibinfo{person}{Daniela Rus}} (Eds.).
  \bibinfo{publisher}{Springer}, \bibinfo{pages}{157--173}.
\newblock


\bibitem[Yu and LaValle(2013)]%
        {yu2013planning}
\bibfield{author}{\bibinfo{person}{Jingjin Yu} {and} \bibinfo{person}{Steven~M
  LaValle}.} \bibinfo{year}{2013}\natexlab{}.
\newblock \showarticletitle{Planning optimal paths for multiple robots on
  graphs}. In \bibinfo{booktitle}{\emph{2013 IEEE International Conference on
  Robotics and Automation}}. \bibinfo{pages}{3612--3617}.
\newblock


\bibitem[Zhou et~al\mbox{.}(2015)]%
        {zhou2015constraint}
\bibfield{author}{\bibinfo{person}{Neng-Fa Zhou}, \bibinfo{person}{H{\aa}kan
  Kjellerstrand}, {and} \bibinfo{person}{Jonathan Fruhman}.}
  \bibinfo{year}{2015}\natexlab{}.
\newblock \bibinfo{booktitle}{\emph{Constraint solving and planning with
  Picat}}. Vol.~\bibinfo{volume}{11}.
\newblock \bibinfo{publisher}{Springer}.
\newblock


\bibitem[Zhou et~al\mbox{.}(2000)]%
        {ZTN00}
\bibfield{author}{\bibinfo{person}{Xiao Zhou}, \bibinfo{person}{Syurei Tamura},
  {and} \bibinfo{person}{Takao Nishizeki}.} \bibinfo{year}{2000}\natexlab{}.
\newblock \showarticletitle{Finding Edge-Disjoint Paths in Partial
  \emph{k}-Trees}.
\newblock \bibinfo{journal}{\emph{Algorithmica}} \bibinfo{volume}{26},
  \bibinfo{number}{1} (\bibinfo{year}{2000}), \bibinfo{pages}{3--30}.
\newblock
\urldef\tempurl%
\url{https://doi.org/10.1007/S004539910002}
\showDOI{\tempurl}


\end{thebibliography}

\end{document}